%% file: main.tex
\RequirePackage{etex}
\documentclass[11pt]{article}
\usepackage[margin=.95in]{geometry} 
\usepackage[T1]{fontenc}    
\usepackage{lmodern}
\usepackage{nicefrac}       
\usepackage{refcount}
\usepackage{amssymb}
\usepackage[LGR]{fontenc}
\usepackage[greek,english]{babel}
\usepackage{tikz}
\usetikzlibrary{arrows.meta,positioning,shapes.geometric,calc}

\usepackage{tcolorbox}

\usepackage{fullpage}
\usepackage{authblk}
\usepackage{complexity}
\usepackage{nicefrac}
\usepackage{gensymb}

\usepackage{enumitem}
\usepackage{microtype}
\usepackage{graphicx}
\usepackage{subfigure}
\usepackage{booktabs} 
\usepackage{pifont}

\usepackage{tikz}
\usetikzlibrary{positioning}

\usepackage{bbm}

\usepackage{natbib}

\newcommand{\CLS}{\textsf{CLS}\xspace}

\renewcommand{\PLS}{\textsf{PLS}\xspace}

\renewcommand{\PPAD}{\textsf{PPAD}\xspace}
\renewcommand{\PCP}{\textsf{PCP}\xspace}

\newcommand{\ind}{\mathbf{1}}

\newcommand{\PureCircuit}{\textsc{Pure-Circuit}}
\newcommand{\GeneralizedCircuit}{\textsc{Generalized-Circuit}}
\newcommand{\NOT}{\textsc{NOT}}
\newcommand{\OR}{\textsc{OR}}
\newcommand{\PURIFY}{\textsc{PURIFY}}
\newcommand{\Controller}{C}
\newcommand{\RewardPlayer}{R}

\newcommand{\PlayerSet}{[m]}
\newcommand{\StateSpace}{\mathcal{S}}
\newcommand{\NodeSet}{\mathcal{V}}

\newcommand{\GateFamily}{\mathcal H}
\newcommand{\GateSet}{\GateFamily}

\newcommand{\InitialDist}{\mu}

\newcommand{\ActionSymbol}{\mathcal{A}}
\newcommand{\ActionSet}[2]{\ActionSymbol_{#1}(#2)}
\newcommand{\ActionProfileSet}[1]{\ActionSymbol(#1)}
\newcommand{\Action}[1]{a_{#1}}
\newcommand{\RewardSymbol}{r}

\newcommand{\TransKernel}{{P}}
\newcommand{\Simplex}[1]{\Delta(#1)}
\newcommand{\Discount}{\gamma}
\newcommand{\Approx}{\varepsilon}
\newcommand{\Policy}{\pi}

\newcommand{\OtherPolicy}[1]{\Policy_{-#1}}
\newcommand{\DeviationPolicy}[1]{\Policy'_{#1}}
\newcommand{\DeviationProfile}[1]{\DeviationPolicy{#1}\times\OtherPolicy{#1}}

\newcommand{\Expect}{\mathbb E}
\newcommand{\Prob}{\mathbb P}

\newcommand{\ValSym}{V}
\newcommand{\QSym}{Q}
\newcommand{\AvgQSym}{\overline Q}
\newcommand{\AdvSym}{\mathsf{Adv}}

\newcommand{\StageRegOp}{\rho}
\newcommand{\Value}[3]{\ValSym_{#1}^{#2}(#3)}

\newcommand{\QValue}[3]{\QSym_{#1}^{#2}(#3)}
\newcommand{\AvgQValue}[4]{\AvgQSym_{#1}^{#2}(#3,#4)}
\newcommand{\CoarseAdvantage}[4]{\AdvSym_{#1}^{#2,#3}(#4)}
\newcommand{\ActionAdvantage}[4]{\AdvSym_{#1}^{#2}(#3,#4)}

\newcommand{\ControllerBestGap}[1]{\StageRegOp_{\Controller,#1}^{\Policy}}
\newcommand{\VisitDist}[3]{d_{#1}^{#2}(#3)}

\newcommand{\ContVal}[2]{W_{#1}^{\Policy}(#2)}
\newcommand{\FutureVal}[1]{U_{#1}}

\newcommand{\State}[1]{s_{#1}}
\newcommand{\Sink}[1]{z_{#1}}
\newcommand{\CtrlMarg}[1]{q_{#1}}
\newcommand{\SignalMarg}[1]{p_{#1}}
\newcommand{\RewMarg}[1]{\SignalMarg{#1}}

\newcommand{\MatchProb}[1]{M_{#1}}

\newcommand{\Decoded}[1]{x[#1]}

\newcommand{\Bits}{\{0,1\}}
\newcommand{\BoolOrUnknown}{\{0,1,\bot\}}
\newcommand{\LowCutoff}{\frac14}
\newcommand{\HighCutoff}{\frac34}
\newcommand{\FixedDiscount}{\frac1{16}}
\newcommand{\HighCtrlThreshold}{\frac{15}{16}}
\newcommand{\SmallGap}{\frac1{256}}

\newcommand{\Gates}{\mathcal H}
\newcommand{\OutputStates}{\mathcal O}
\newcommand{\ProdMatch}[1]{M^{\mathrm{prod}}_{#1}}
\newcommand{\LowCtrlThreshold}{\frac1{16}}
\newcommand{\GateGap}{\Delta}
\newcommand{\Copy}{\mathcal{E}_{\mathrm{copy}}}
\newcommand{\Force}{\mathcal{E}_{\mathrm{force}}}
\newcommand{\Bad}{\mathsf{Bad}}

\usepackage{url}
\usepackage{amsmath}
\usepackage{amssymb}
\usepackage{amsthm}
\usepackage{bbm}
\usepackage{algorithm}
\usepackage[noend]{algorithmic}
\usepackage[algo2e,vlined,ruled,linesnumbered]{algorithm2e}
\usepackage{footmisc}
\usepackage[table]{xcolor}
\makeatletter
\newcommand\footnoteref[1]{\protected@xdef\@thefnmark{\ref{#1}}\@footnotemark}
\makeatother
\usepackage{arydshln}

\SetAlFnt{\small}
\SetAlCapFnt{\small}
\SetAlCapNameFnt{\small}
\SetAlCapHSkip{0pt}
\IncMargin{-\parindent}

\usepackage[hidelinks, colorlinks = true, linkcolor=red, citecolor=red]{hyperref}
\usepackage{cleveref}
\newcommand{\declarecolor}[2]{\definecolor{#1}{RGB}{#2}\expandafter\newcommand\csname #1\endcsname[1]{\textcolor{#1}{##1}}}
\declarecolor{White}{255, 255, 255}
\declarecolor{Black}{0, 0, 0}
\declarecolor{Maroon}{128, 0, 0}
\declarecolor{Coral}{255, 127, 80}
\declarecolor{Red}{182, 21, 21}
\declarecolor{LimeGreen}{50, 205, 50}
\declarecolor{DarkGreen}{0, 80, 0}
\declarecolor{Purple}{146, 42, 158}
\declarecolor{Navy}{0, 0, 128}
\declarecolor{LightBlue}{84, 101, 202}
\definecolor{mydarkblue}{rgb}{0,0.08,0.45}
\usepackage{tcolorbox}

\usepackage{xcolor}	

\definecolor{CardinalRed}{HTML}{C41E3A}
\definecolor{Dartmouth}{HTML}{00693E}
\definecolor{SapphireBlue}{HTML}{0F52BA}

\colorlet{MyRed}{CardinalRed}
\colorlet{MyGreen}{Dartmouth}

\colorlet{MyLightRed}{MyRed!25}
\colorlet{MyLightGreen}{MyGreen!25}

\colorlet{AlertColor}{MyRed}	
\colorlet{BadColor}{MyRed}	
\colorlet{FocusColor}{MyRed}	
\colorlet{GoodColor}{MyGreen}	
\colorlet{MacroColor}{MyRed}	

\hypersetup{ %
    pdftitle={},
    pdfkeywords={},
    pdfborder=0 0 0,
    pdfpagemode=UseNone,
    colorlinks=true,
    linkcolor=MyRed,
    citecolor=MyGreen,
    filecolor=Purple,
    urlcolor=blue,
}

\usepackage[%
linewidth=2pt,
linecolor=gray,
middlelinecolor= black,
middlelinewidth=0.4pt,
roundcorner=1pt,
topline = false,
rightline = false,
bottomline = false,
rightmargin=0pt,
skipabove=0pt,
skipbelow=0pt,
leftmargin=0pt,
innerleftmargin=4pt,
innerrightmargin=0pt,
innertopmargin=0pt,
innerbottommargin=0pt,
]{mdframed}

\definecolor{TheoremBlueFrame}{RGB}{25,92,170}
\definecolor{LemmaGreenFrame}{RGB}{46,125,50}

\definecolor{TheoremBlueBack}{RGB}{245,249,255}
\definecolor{LemmaGreenBack}{RGB}{247,252,247}

\newenvironment{whitebox}
  {\begin{tcolorbox}[
    colframe=black,
    colback=white,
    top=2pt,
    left=2pt,
    right=2pt,
    bottom=2pt
  ]}
  {\end{tcolorbox}}

\newenvironment{gadgetbox}
  {\begin{tcolorbox}[
    colframe=MyRed,
    colback=red!5,
    top=2pt,
    left=2pt,
    right=2pt,
    bottom=2pt
  ]}
  {\end{tcolorbox}}

\newenvironment{statementbox}
  {\begin{tcolorbox}[
    colframe=TheoremBlueFrame,
    colback=TheoremBlueBack,
    top=2pt,
    left=2pt,
    right=2pt,
    bottom=2pt
  ]}
  {\end{tcolorbox}}

\let\oldnl\nl
\newcommand{\nonl}{\renewcommand{\nl}{\let\nl\oldnl}}

\newcounter{protocol}


\DeclareMathOperator*{\argmax}{arg\,max}

\usepackage{array}
\usepackage[utf8]{inputenc} 
\usepackage[T1]{fontenc}    
\usepackage{booktabs}       
\usepackage{amsfonts}       
\usepackage{nicefrac}       
\usepackage{microtype}      
\usepackage{amsmath}
\usepackage{amssymb}
\usepackage{mathtools}
\usepackage{amsthm}
\usepackage{enumitem}

\usepackage{multirow}
\usepackage{xspace}
\usepackage{xurl} 

\newcommand{\gadgetC}{{\color{MyRed}(\textbf{C})}\xspace}
\newcommand{\gadgetF}{{\color{MyRed}(\textbf{F})}\xspace}

\theoremstyle{plain}
\newtheorem*{theorem*}{Theorem}
\newtheorem{theorem}{Theorem}[section]
\newtheorem{proposition}[theorem]{Proposition}
\newtheorem{lemma}[theorem]{Lemma}
\newtheorem*{lemma*}{Lemma}

\newtheorem*{corollary*}{Corollary}
\theoremstyle{definition}
\newtheorem{definition}[theorem]{Definition}
\newtheorem*{definition*}{Definition}

\newtheorem{remark}[theorem]{Remark}
\theoremstyle{observation}

\newtheorem{conjecture}{Conjecture}

\definecolor{actionzero}{RGB}{230,145,56}
\definecolor{actionone}{RGB}{61,133,198}
\definecolor{ownerA}{RGB}{204,0,0}
\definecolor{ownerBfill}{RGB}{220,220,220}
\definecolor{sinkfill}{RGB}{245,245,245}

\tikzset{
  Astate/.style={draw=ownerA, very thick, rectangle, minimum width=10mm, minimum height=8mm, inner sep=1pt},
  Bstate/.style={draw=black, fill=ownerBfill, very thick, circle, minimum size=8mm, inner sep=0pt},
  sinkstate/.style={draw=black, fill=sinkfill, rounded corners=2pt, minimum width=14mm, minimum height=8mm, inner sep=1pt},
  zeropath/.style={-{Latex[length=2mm]}, very thick, actionzero},
  onepath/.style={-{Latex[length=2mm]}, very thick, actionone},
  lab/.style={font=\small}
}

\usepackage[textsize=tiny]{todonotes}

\usepackage{dsfont}

\allowdisplaybreaks

\title{\fontsize{18}{22}\selectfont The Complexity of Computing Coarse Correlated Equilibria\\
in Markov Games with a Single Controller}

\author[1]{Gabriele Farina}
\author[2,3]{Andreas Kontogiannis\textsuperscript{$\dagger$}}
\author[4]{Ioannis Panageas}
\author[3,5]{Vasilis Pollatos\textsuperscript{$\dagger$}}

\affil[1]{Massachusetts Institute of Technology}
\affil[2]{National Technical University of Athens}
\affil[3]{Archimedes, Athena Research Center, Greece}
\affil[4]{University of California, Irvine}
\affil[5]{National and Kapodistrian University of Athens}

\begin{document}

\maketitle

\begingroup
\renewcommand{\thefootnote}{\fnsymbol{footnote}}
\footnotetext[2]{%
Part of this work was completed during a research visit to University of California, Irvine.\\
Corresponding authors:
\nolinkurl{gfarina@mit.edu},
\nolinkurl{andreaskontogiannis@mail.ntua.gr},
\nolinkurl{ipanagea@ics.uci.edu},
\nolinkurl{vaspoll@math.uoa.gr}.%
}
\endgroup

\setcounter{footnote}{0}
\renewcommand{\thefootnote}{\arabic{footnote}}
\begin{abstract}

\input{abstract}

\end{abstract}
\thispagestyle{empty}
\newpage

\tableofcontents
\thispagestyle{empty}
\clearpage
\setcounter{page}{1}
\newpage

\input{introduction}

\section{Preliminaries}\label{sec: prelims}

\subsection{Markov games, stationary policies, and value functions}

\begin{definition}[Markov game]
An infinite-horizon discounted Markov game is a tuple
\[
\Gamma=(\StateSpace,\PlayerSet,(\ActionSet{i}{s})_{i\in\PlayerSet,s\in\StateSpace},(\RewardSymbol_i)_{i\in\PlayerSet},\TransKernel,\InitialDist,\Discount),
\]
where $\StateSpace$ is a finite state space, $\PlayerSet$ is a finite set of players, $\ActionSet{i}{s}$ is the action set of player $i$ at state $s$, and $\ActionProfileSet{s}=\prod_{i\in\PlayerSet}\ActionSet{i}{s}$.  The reward function of player $i$ is
$\RewardSymbol_i:\bigcup_{s\in\StateSpace}\{s\}\times\ActionProfileSet{s}\to[0,1]$.
For every $s$ and $a\in\ActionProfileSet{s}$, the transition kernel is $\TransKernel(\cdot\mid s,a)\in\Simplex{\StateSpace}$; $\InitialDist\in\Simplex{\StateSpace}$ is the initial distribution; and $\Discount\in[0,1)$ is the discount factor.
\end{definition}

A stationary Markov joint policy is a map $\Policy:\StateSpace\to\Simplex{\ActionProfileSet{s}}$.  We write $\Policy(a\mid s)$ for the probability of joint action $a$ at $s$.  The normalized value function of player $i$ from state $s$ is
\begin{equation}
\Value{i}{\Policy}{s}
=(1-\Discount)\Expect_s^{\Policy}\left[\sum_{t\ge0}\Discount^t \RewardSymbol_i(s_t,a_t)\right].
\label{eq:value}
\end{equation}
All values lie in $[0,1]$.  The action-value function is
\begin{equation}
\QValue{i}{\Policy}{s,a}
=(1-\Discount)\RewardSymbol_i(s,a)+\Discount\sum_{s'\in\StateSpace}\TransKernel(s'\mid s,a)\Value{i}{\Policy}{s'}.
\label{eq:qvalue}
\end{equation}
For an initial distribution $\mu\in\Simplex{\StateSpace}$, we abbreviate the
corresponding expected value by
\[
\Value{i}{\Policy}{\mu}
:=
\Expect_{s_0\sim\mu}\!\left[\Value{i}{\Policy}{s_0}\right]
=
\sum_{s\in\StateSpace}\mu(s)\Value{i}{\Policy}{s}.
\]

\begin{definition}[\textit{Single-controller Markov game}]
A Markov game is single-controller if there is a specific player $\Controller\in\PlayerSet$ such that, for every state $s$ and action profile $a$, the transition kernel $\TransKernel(\cdot\mid s,a)$ depends on $a$ only through $a_{\Controller}$.
\end{definition}

\subsection{Markov Coarse Correlated Equilibria}

At every visit to $s$, the correlation device samples a full action profile $a\sim\Policy(\cdot\mid s)$.  We write $\OtherPolicy{i}(\cdot\mid s)$ for the marginal over the actions of all players except $i$.  A stationary one-player deviation for player $i$ is a policy $\DeviationPolicy{i}:\StateSpace\to\Simplex{\ActionSet{i}{s}}$.  The deviated joint policy $\DeviationProfile{i}$ samples $b_i\sim\DeviationPolicy{i}(\cdot\mid s)$ independently of $a_{-i}\sim\OtherPolicy{i}(\cdot\mid s)$:
\begin{equation}
(\DeviationProfile{i})(b_i,a_{-i}\mid s)
=
\DeviationPolicy{i}(b_i\mid s)\sum_{a_i\in\ActionSet{i}{s}}\Policy(a_i,a_{-i}\mid s).
\label{eq:deviation-product-policy}
\end{equation}
A deterministic stationary coarse deviation is the special case in which $\DeviationPolicy{i}(b_i\mid s)=\ind[b_i=b_i^\star(s)]$ for a state-dependent map $b_i^\star$.

\begin{definition}[Perfect Markov $\Approx$-CCE]
A stationary Markov joint policy $\Policy$ is a perfect Markov $\Approx$-coarse correlated equilibrium if, for every player $i$, every state $s$, and every stationary one-player deviation $\DeviationPolicy{i}$,
\[
\Value{i}{\DeviationProfile{i}}{s}-\Value{i}{\Policy}{s}\le \Approx.
\]
\end{definition}

\begin{definition}[Non-perfect Markov $\Approx$-CCE]
A stationary Markov joint policy $\Policy$ is a non-perfect Markov $\Approx$-coarse correlated equilibrium from $\InitialDist$ if, for every player $i$ and every stationary one-player deviation $\DeviationPolicy{i}$,
\[
\Value{i}{\DeviationProfile{i}}{\InitialDist}-\Value{i}{\Policy}{\InitialDist}\le \Approx.
\]
\end{definition}

For a stationary policy profile $\pi$, define the normalized discounted state-visitation distribution from $\mu$ by
\[
\VisitDist{\mu}{\pi}{s}:=(1-\Discount)\sum_{t\ge0}\Discount^t\Pr^{\pi}[s_t=s\mid s_0\sim\mu].
\]
Moreover, it is easy to see that
\begin{equation}
\VisitDist{\mu}{\pi}{s}\ge (1-\Discount)\mu(s)
\quad\text{for every }s \in \StateSpace.
\label{eq:visit-lower}
\end{equation}

\subsection{Performance-difference lemma for one-player deviations}

For a stationary joint policy $\Policy$ and a stationary one-player deviation
policy $\DeviationPolicy{i}$, define the local deviation advantage by
\[
\CoarseAdvantage{i}{\Policy}{\DeviationPolicy{i}}{s}
:=
\sum_{b_i,a_{-i}}
(\DeviationProfile{i})(b_i,a_{-i}\mid s)
\QValue{i}{\Policy}{s,(b_i,a_{-i})}
-
\Value{i}{\Policy}{s}.
\]

Let
$\VisitDist{\State{0}}{\DeviationProfile{i}}{s}$ be the normalized discounted
state-visitation distribution under $\DeviationProfile{i}$ starting from
$\State{0}$:
\[
\VisitDist{\State{0}}{\DeviationProfile{i}}{s}
:=
(1-\Discount)\sum_{t\ge0}\Discount^t
\Prob[\State{t}=s\mid \State{0},\DeviationProfile{i}].
\]

Next, we provide the performance-difference lemma \cite{kakadeBook,agarwal2021theory} variant for one-player deviations:

\begin{lemma}[Performance-difference  lemma for one-player deviations]
\label{lem:performance-difference-deviation-policy}
Fix a stationary joint policy $\Policy$ and a deviation
policy $\DeviationPolicy{i}$ for player $i$.  
Then
\begin{equation}
\Value{i}{\DeviationProfile{i}}{\State{0}}-\Value{i}{\Policy}{\State{0}}
=
\frac{1}{1-\Discount}
\sum_{s\in\StateSpace}
\VisitDist{\State{0}}{\DeviationProfile{i}}{s}
\CoarseAdvantage{i}{\Policy}{\DeviationPolicy{i}}{s}.
\label{eq:coarse-perf-diff}
\end{equation}
\end{lemma}
\noindent
For the interested reader, we defer the proof of Lemma \ref{lem:performance-difference-deviation-policy} to Appendix \ref{appendix:performance}.

\medskip

\subsection{\PureCircuit{} and the \PCP-for-\PPAD hypothesis}

Next, we provide the formal definition of \PureCircuit{}, the \PPAD-complete problem that we will use to prove our hardness results. 

\begin{whitebox}
\begin{definition}[\PureCircuit{} \cite{purecircuit}]
A \textsc{Pure-Circuit} instance consists of a node set $\NodeSet$ and a gate set $\GateSet$ 
of types \NOT, \OR, and \PURIFY. 
Each node is the output node of at most one gate. A solution is an assignment
$
x\in\BoolOrUnknown^V
$
satisfying:
\begin{itemize}[leftmargin=7mm]
\item $\NOT(u,v)$:
\[
x[u]=0\Rightarrow x[v]=1,
\qquad
x[u]=1\Rightarrow x[v]=0.
\]
\item $\OR(u_0,u_1,w)$:
\[
x[u_0]=x[u_1]=0\Rightarrow x[w]=0,
\]
and
\[
x[u_0]=1\ \text{or}\ x[u_1]=1\Rightarrow x[w]=1.
\]
\item $\PURIFY(u,v_0,v_1)$:
\[
\{x[v_0],x[v_1]\}\cap\Bits\neq\emptyset,
\]
and if $x[u]\in\Bits$, then
\[
x[v_0]=x[v_1]=x[u].
\]
\end{itemize}
\end{definition}
\end{whitebox}

\PureCircuit{} is \PPAD-complete for this gate basis; output uniqueness is used only to ensure that one circuit node receives at most one transition rule in the reduction.

\begin{definition}[$\delta$-\PureCircuit{}]
For $\delta\in(0,1)$, the problem $\delta$-\PureCircuit{} asks, given a \PureCircuit{} instance, to find an assignment $x:\NodeSet\to\BoolOrUnknown$ satisfying at least a $(1-\delta)$ fraction of its gates.
\end{definition}

\begin{conjecture}[\PCP-for-\PPAD, \PureCircuit{} version]
\label[conjecture]{conj:pcp-ppad}
There exists a constant $\delta_{\mathrm{pcp}}>0$ such that $\delta_{\mathrm{pcp}}$-\PureCircuit{} is \PPAD-hard, even when every node is the input to exactly one gate.
\end{conjecture}

This is the \PureCircuit{} formulation of the \PCP-for-\PPAD conjecture~\cite{babichenko2016can} stated in \cite{deligkas2026fisher}.  The constant-accuracy non-perfect theorem in Section \ref{sec:non-perfect} is conditional on \Cref{conj:pcp-ppad}; the inverse-polynomial-accuracy non-perfect theorem and the perfect theorem (Section \ref{sec:perfect}) are unconditional.

\section{The single-controller game construction}
\label{sec:construction}

\subsection{Game definition}

Fix $\Discount=\FixedDiscount$.  Given a \PureCircuit{} instance $I$, construct a two-player discounted Markov game $G_\lambda(I)$ with players $\Controller$ and $\RewardPlayer$; 
$\Controller$ (\textit{controller}) is the player who controls the transitions; $\RewardPlayer$ is called the \textit{reward player}.
The parameter $\lambda\in[0,1)$ controls the controller's decorrelation bonus, as we will next define formally. 
The reduction for perfect Markov CCE uses $\lambda=0$; the non-perfect reduction uses a fixed $\lambda \in (0,1/180)$.
A graphical illustration of our construction can be found in Figure \ref{fig:gates}.

\paragraph{Variable states and marginals.}
For every node $u\in\NodeSet$, create a variable state $\State{u}$.  At $\State{u}$, both players have action set $\Bits$.  For a stationary joint policy $\Policy$, define
\begin{equation}
\CtrlMarg{u}:=\sum_{a\in\ActionProfileSet{\State{u}}:a_{\Controller}=1}\Policy(a\mid\State{u}),
\qquad
\RewMarg{u}:=\sum_{a\in\ActionProfileSet{\State{u}}:a_{\RewardPlayer}=1}\Policy(a\mid\State{u}),
\label{eq:pq-def}
\end{equation}
and the matching probability
\begin{equation}
\MatchProb{u}:=\sum_{a\in\ActionProfileSet{\State{u}}:a_{\Controller}=a_{\RewardPlayer}}\Policy(a\mid\State{u}).
\label{eq:M-def}
\end{equation}
We use the following decoded circuit assignment:
\begin{equation}
x[u]=
\begin{cases}
0,&\RewMarg{u}\le\LowCutoff,\\[1mm]
1,&\RewMarg{u}\ge\HighCutoff,\\[1mm]
\bot,&\text{otherwise.}
\end{cases}
\label{eq:decode}
\end{equation}

\paragraph{Sink states.}
For every rational number $\theta$ used in the gadgets, create an absorbing sink $\Sink{\theta}$.  At $\Sink{\theta}$ both players have a unique dummy action, the state self-loops, the controller receives reward $\theta$, and the reward player receives reward $0$.  Hence
\[
\Value{\Controller}{\Policy}{\Sink{\theta}}=\theta
\quad\text{and}\quad
\Value{\RewardPlayer}{\Policy}{\Sink{\theta}}=0
\]
for every stationary policy $\Policy$.

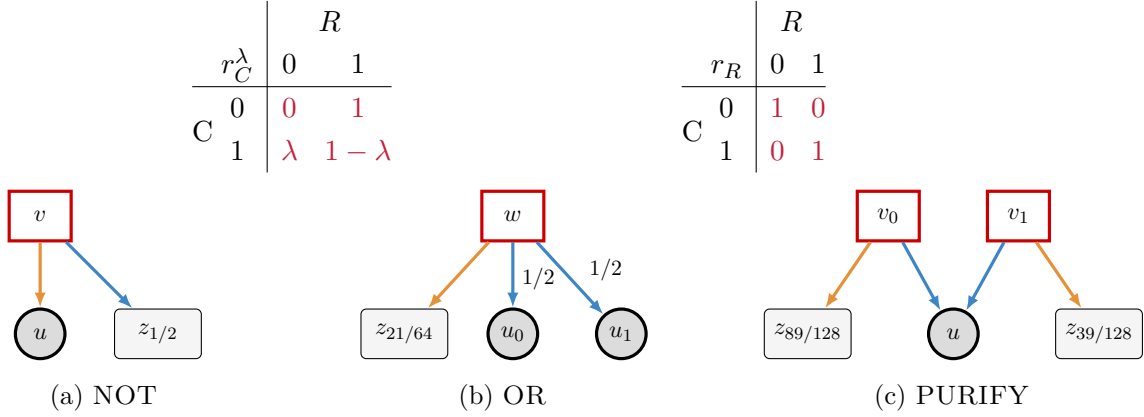
\begin{figure}[t]
\centering
\makebox[\textwidth][l]{%
\hspace*{0.12\textwidth}%
\begin{minipage}{0.34\textwidth}
\centering
\[
\renewcommand{\arraystretch}{1.18}
\begin{array}{@{}c@{\;}c|cc@{}}
\multicolumn{2}{c|}{}&\multicolumn{2}{c}{\RewardPlayer}\\
& r_{\Controller}^{\lambda} & 0 & 1\\ \hline
\raisebox{-0.65\normalbaselineskip}[0pt][0pt]{\Controller} & 0 & {\color{MyRed}0} & {\color{MyRed}1}\\
& 1 & {\color{MyRed}\lambda} & {\color{MyRed}1-\lambda}
\end{array}
\]
\end{minipage}%
\hspace*{0.05\textwidth}%
\begin{minipage}{0.30\textwidth}
\centering
\[
\renewcommand{\arraystretch}{1.18}
\begin{array}{@{}c@{\;}c|cc@{}}
\multicolumn{2}{c|}{}&\multicolumn{2}{c}{\RewardPlayer}\\
& r_{\RewardPlayer} & 0 & 1\\ \hline
\raisebox{-0.65\normalbaselineskip}[0pt][0pt]{\Controller} & 0 & {\color{MyRed}1} & {\color{MyRed}0}\\
& 1 & {\color{MyRed}0} & {\color{MyRed}1}
\end{array}
\]
\end{minipage}%
}

\vspace{.5em}

\begin{minipage}{0.28\textwidth}
\centering
\begin{tikzpicture}[scale=0.82,transform shape]
\node[Astate] (v) at (0,0) {$v$};
\node[Bstate] (u) at (0,-1.9) {$u$};
\node[sinkstate] (z) at (1.9,-1.9) {$\Sink{1/2}$};
\draw[zeropath] (v) -- (u);
\draw[onepath] (v) -- (z);
\end{tikzpicture}

\smallskip
{\small (a) \NOT}
\end{minipage}
\hfill
\begin{minipage}{0.32\textwidth}
\centering
\begin{tikzpicture}[scale=0.82,transform shape]
\node[Astate] (w) at (0,0) {$w$};
\node[sinkstate] (z) at (-1.75,-1.9) {$\Sink{21/64}$};
\node[Bstate] (u0) at (0,-1.9) {$u_0$};
\node[Bstate] (u1) at (1.75,-1.9) {$u_1$};
\draw[zeropath] (w) -- (z);
\draw[onepath] (w) -- node[pos=0.55,right,lab,black] {$1/2$} (u0);
\draw[onepath] (w) -- node[pos=0.65,above right,lab,black] {$1/2$} (u1);
\end{tikzpicture}

\smallskip
{\small (b) \OR}
\end{minipage}
\hfill
\begin{minipage}{0.36\textwidth}
\centering
\begin{tikzpicture}[scale=0.82,transform shape]
\node[Astate] (v0) at (-1.05,0) {$v_0$};
\node[Astate] (v1) at (1.05,0) {$v_1$};
\node[Bstate] (u) at (0,-1.9) {$u$};
\node[sinkstate] (z0) at (-2.35,-1.9) {$\Sink{89/128}$};
\node[sinkstate] (z1) at (2.35,-1.9) {$\Sink{39/128}$};
\draw[zeropath] (v0) -- (z0);
\draw[onepath] (v0) -- (u);
\draw[onepath] (v1) -- (u);
\draw[zeropath] (v1) -- (z1);
\end{tikzpicture}

\smallskip
{\small (c) \PURIFY}
\end{minipage}
\caption{Our single-controller Markov game construction: Player $C$ denotes the transition-controller, and player $R$ denotes the non-controller player. Rectangles are output states, circles are input states, and black rectangles are absorbing sinks.  Orange arrows correspond to controller action $0$ and blue arrows to controller action $1$.  The controller reward table is the $\lambda$-modified reward used in the unified construction; for the perfect theorem, $\lambda=0$. 
}
\label{fig:gates}
\end{figure}

\paragraph{Rewards at variable states.}
At every variable state $\State{u}$, define
\begin{equation}
\RewardSymbol_{\Controller}^{\lambda}(\State{u},a_{\Controller},a_{\RewardPlayer})
=(1-\lambda)a_{\RewardPlayer}+\lambda\ind[a_{\Controller}\ne a_{\RewardPlayer}],
\qquad
\RewardSymbol_{\RewardPlayer}(\State{u},a_{\Controller},a_{\RewardPlayer})
=\ind[a_{\Controller}=a_{\RewardPlayer}].
\label{eq:lambda-reward}
\end{equation}
For $\lambda=0$, the controller reward is just the reward player's signal bit.  
For $\lambda>0$, the second term is the \textit{decorrelation bonus} used only in the non-perfect argument.  
Since $0\le\lambda<1$, both rewards lie in $[0,1]$.

\paragraph{Transitions.}
Only the controller's action $a_{\Controller}$ affects transitions.  If $v$ is the output of $\NOT(u,v)$, then at $\State{v}$,
\[
a_{\Controller}=0\to\State{u},
\qquad
 a_{\Controller}=1\to\Sink{1/2}.
\]
If $w$ is the output of $\OR(u_0,u_1,w)$, then at $\State{w}$,
\[
a_{\Controller}=0\to\Sink{21/64},
\qquad
 a_{\Controller}=1\to
\begin{cases}
\State{u_0},&\text{with probability }1/2,\\
\State{u_1},&\text{with probability }1/2.
\end{cases}
\]
If $v_0,v_1$ are the outputs of $\PURIFY(u,v_0,v_1)$, then at $\State{v_0}$,
\[
a_{\Controller}=0\to\Sink{89/128},
\qquad
 a_{\Controller}=1\to\State{u},
\]
and at $\State{v_1}$,
\[
a_{\Controller}=0\to\Sink{39/128},
\qquad
 a_{\Controller}=1\to\State{u}.
\]
If a node is not the output of any gate, both controller actions transition to $\Sink{0}$.  The output-uniqueness of the source instance ensures that these rules never conflict.

\paragraph{Further useful quantities.}
For a variable state $\State{v}$ and $b\in\Bits$, define the controller's expected continuation value if the current controller action is fixed to $b$ as follows
\begin{equation}
\ContVal{b}{\State{v}}
:=\sum_{s'\in\StateSpace}\TransKernel(s'\mid \State{v},b)\Value{\Controller}{\Policy}{s'}.
\label{eq:W-def}
\end{equation}
Also, define the controller's current action value averaged against the reward player's marginal:
\begin{equation}
\AvgQValue{\Controller}{\Policy}{\State{v}}{b}
:=\sum_{a_{\RewardPlayer}}\Policy_{\RewardPlayer}(a_{\RewardPlayer}\mid\State{v})
\QValue{\Controller}{\Policy}{\State{v},(b,a_{\RewardPlayer})}.
\label{eq:barQ-def}
\end{equation}
For $b\in\Bits$, define the controller's action advantage at a variable state $s$ by
\begin{equation}
\ActionAdvantage{\Controller}{\Policy}{s}{b}
:=\AvgQValue{\Controller}{\Policy}{s}{b}-\Value{\Controller}{\Policy}{s},
\label{eq: adv_b}
\end{equation}
and define the controller's best local advantage
\begin{equation}
\ControllerBestGap{s}:=\max_{b\in\Bits}\ActionAdvantage{\Controller}{\Policy}{s}{b}.
\label{eq: best_advantage}
\end{equation}


\subsection{Single-controller gadgets for implementing the gates}
\label{sec:gadgets}

This section will serve as the key theoretical framework for proving our hardness results. 
It does not use a particular equilibrium definition directly. Instead it assumes the following two gadget conditions, which are proved in the case-specific sections.
In particular, we will prove these conditions holds under both perfect and non-perfect Markov $\Approx$-CCE. 

Formally, our analysis leverages the following two conditions of a stationary policy $\Policy$ on the constructed game:

\begin{gadgetbox}
\paragraph{{\color{MyRed}(\textbf{C})} Copy gadget.}\phantomsection\label{copy}  For every variable state $\State{u}$, it holds that
\[
\CtrlMarg{u}\ge\HighCtrlThreshold
\Rightarrow
\RewMarg{u}>\HighCutoff,
\qquad
\CtrlMarg{u}\le\FixedDiscount
\Rightarrow
\RewMarg{u}<\LowCutoff.
\]
\paragraph{{\color{MyRed}(\textbf{F})} Controller-forcing gadget.}\phantomsection\label{controller-forcing}  For every output variable state
$\State{v}$ (i.e., a state $s_v$ corresponding to a node $v$ which is an output node of some gate) and gap $\Delta:=\SmallGap$, it holds that
\[
\ContVal{1}{\State{v}}-\ContVal{0}{\State{v}}\ge\Delta
\Rightarrow
\CtrlMarg{v}\ge\HighCtrlThreshold,
\]
and
\[
\ContVal{0}{\State{v}}-\ContVal{1}{\State{v}}\ge\Delta
\Rightarrow
\CtrlMarg{v}\le\FixedDiscount.
\]
\end{gadgetbox}

\paragraph{High-level idea.}
The copy gadget~\hyperref[copy]{\gadgetC} transfers the controller's almost-Boolean choice at a variable state to the reward player's marginal action. 
At state $\State{u}$, the reward player $\RewardPlayer$ receives payoff $1$ exactly when its action matches the controller's action. 
Hence, if the controller plays action $1$ with very high probability but the reward player does not play $1$ with high probability, then the reward player has a large immediate gain from the coarse deviation that plays $1$ at $\State{u}$.
Similarly, if the controller plays action $1$ with very low probability but the reward player does not play $0$ with high probability, then the reward player has a large immediate gain from the coarse deviation that plays $0$ at $\State{u}$. 
Since the reward-player deviations do not affect transitions, the worst possible continuation loss is at most $\Discount$. 
With our choice $\Discount=1/16$, the immediate gain in either bad case dominates this continuation loss.

The controller-forcing gadget~\hyperref[controller-forcing]{\gadgetF} is the mechanism that makes each gate compute its output. 
When $\lambda=0$, at an output state $\State{v}$, the controller's immediate reward is independent of the controller's own action; it depends only on the reward player's action. 
Therefore, in this case, the controller's incentive between actions $0$ and $1$ comes only from the discounted continuation values $\ContVal{0}{\State{v}}$ and $\ContVal{1}{\State{v}}$. 
For example, if one action has continuation value larger by a fixed gap, then a perfect Markov CCE can put only a small probability on the worse action; otherwise the controller would have a profitable stationary coarse deviation to the better fixed action. 
Hence a continuation-value gap forces $\CtrlMarg{v}$ close to $0$ or close to $1$, and the copy gadget then turns this into a rounded output signal $\RewMarg{v}$.

Next, we provide a useful lemma which effectively bounds the value of the controller at any state of our construction:

\begin{lemma}[Controller value interval]
\label{lem:controller-value}
For every variable state $\State{u}$,
\begin{equation}
(1-\Discount)(1-\lambda)\RewMarg{u}
\le
\Value{\Controller}{\Policy}{\State{u}}
\le
(1-\Discount)\bigl((1-\lambda)\RewMarg{u}+\lambda\bigr)+\Discount.
\label{eq:controller-value-interval}
\end{equation}
In particular, when $\Discount=1/16$, the following hold:
\begin{align}
\RewMarg{u}\le\frac14&\Rightarrow \Value{\Controller}{\Policy}{\State{u}}\le \frac{19}{64}+\frac{45\lambda}{64},
\label{eq:read-low-lambda}\\
\RewMarg{u}\ge\frac34&\Rightarrow \Value{\Controller}{\Policy}{\State{u}}\ge \frac{45}{64}-\frac{45\lambda}{64},
\label{eq:read-high-lambda}\\
\RewMarg{u}\le\frac23&\Rightarrow \Value{\Controller}{\Policy}{\State{u}}\le \frac{88}{128}+\frac{40\lambda}{128},
\label{eq:read-two-thirds-lambda}\\
\RewMarg{u}\ge\frac13&\Rightarrow \Value{\Controller}{\Policy}{\State{u}}\ge \frac{40}{128}-\frac{40\lambda}{128}.
\label{eq:read-one-third-lambda}
\end{align}
\end{lemma}

\begin{proof}
Let $\FutureVal{u}:=\Expect[\Value{\Controller}{\Policy}{\State{1}}\mid \State{0}=\State{u}]$ be the expected next-state value of the controller under $\Policy$ after leaving $\State{u}$.
At $\State{u}$, the expected immediate reward of the controller is
\[
\Expect_{a\sim\Policy(\cdot\mid\State{u})}\left[(1-\lambda)a_{\RewardPlayer}+\lambda\ind[a_{\Controller}\ne a_{\RewardPlayer}]\right]
=(1-\lambda)\RewMarg{u}+\lambda(1-\MatchProb{u}).
\]
By the Bellman equation, we have
\[
\Value{\Controller}{\Policy}{\State{u}}
=(1-\Discount)\bigl((1-\lambda)\RewMarg{u}+\lambda(1-\MatchProb{u})\bigr)+\Discount U_u,
\]  
Since all values lie in $[0,1]$ and $1-\MatchProb{u}\in[0,1]$, we have $0\le U_u\le1$, which proves \eqref{eq:controller-value-interval}.  Substituting $\Discount=1/16$ and the stated bounds on $\RewMarg{u}$ gives the four displayed inequalities.
\end{proof}


Next, we implement the gates NOT, OR, and PURIFY assuming that the two gadget conditions, restricted to the input and output states of the gate under consideration, hold.  

\paragraph{Parameter choices.}
The sink values are chosen to sit strictly between the controller value ranges corresponding to the relevant decoded input regimes, as certified by Lemma~\ref{lem:controller-value}. The NOT sink $1/2$ separates the low-input and high-input value intervals, yielding the two gaps $(13-45\lambda)/64$ in Proposition~\ref{prop:gate-soundness}. The OR sink $21/64$ separates the all-low case from the case in which at least one input is high, yielding gaps $(2-45\lambda)/64$ and $(3-45\lambda)/128$. The PURIFY sinks $89/128$ and $39/128$ are placed around the value thresholds corresponding to $\RewMarg{u}=2/3,3/4,1/4,$ and $1/3$, yielding gaps $(1-40\lambda)/128$ and $(1-90\lambda)/128$. The tight constraints are the PURIFY inequalities with gap $(1-90\lambda)/128$; requiring this to exceed $\GateGap=1/256$ gives
\[
\frac{1-90\lambda}{128}>\frac{1}{256}
\quad\Longleftrightarrow\quad
\lambda<\frac{1}{180}.
\]
All other gate inequalities have larger slack under this choice.

\bigskip

\begin{proposition}[Gates Implementation]
\label{prop:gate-soundness}
Assume $0\le\lambda<1/180$.  If gadget conditions~\hyperref[copy]{\gadgetC} and~\hyperref[controller-forcing]{\gadgetF} hold at all output states of a gate $g \in \{\NOT{}, \OR{}, \PURIFY{} \}$, then the assignment decoded by \eqref{eq:decode} satisfies that gate.
\end{proposition}

\begin{proof}
We verify the three gate types as follows:

\paragraph{\NOT{}.}
Let $g=\NOT(u,v)$.  The transition rules give
\[
\ContVal{0}{\State{v}}=\Value{\Controller}{\Policy}{\State{u}},
\qquad
\ContVal{1}{\State{v}}=\Value{\Controller}{\Policy}{\Sink{1/2}}=\frac12.
\]
We examine the following cases:
\begin{itemize}
\item 
If $x[u]=0$, then $\RewMarg{u}\le1/4$.  By \eqref{eq:read-low-lambda}, we have
\[
\ContVal{1}{\State{v}}-\ContVal{0}{\State{v}}
\ge
\frac12-\left(\frac{19}{64}+\frac{45\lambda}{64}\right)
=
\frac{13-45\lambda}{64}.
\]
Since $\lambda<1/180$, this is strictly larger than $1/256=\GateGap$.  By condition~\hyperref[controller-forcing]{\gadgetF}, we have $\CtrlMarg{v}\ge15/16$, and then condition~\hyperref[copy]{\gadgetC} gives $\RewMarg{v}>3/4$.  Hence $x[v]=1$.
\item
If $x[u]=1$, then $\RewMarg{u}\ge3/4$.  By \eqref{eq:read-high-lambda}, we have 
\[
\ContVal{0}{\State{v}}-\ContVal{1}{\State{v}}
\ge
\left(\frac{45}{64}-\frac{45\lambda}{64}\right)-\frac12
=
\frac{13-45\lambda}{64}
>
\GateGap.
\]
Condition~\hyperref[controller-forcing]{\gadgetF} gives $\CtrlMarg{v}\le1/16$, and condition~\hyperref[copy]{\gadgetC} gives $\RewMarg{v}<1/4$.  Hence $x[v]=0$.
\end{itemize}

\paragraph{\OR{}.}
Let $g=\OR(u_0,u_1,w)$.  Then, we have
\[
\ContVal{0}{\State{w}}=\frac{21}{64},
\qquad
\ContVal{1}{\State{w}}
=\frac12\Value{\Controller}{\Policy}{\State{u_0}}
+\frac12\Value{\Controller}{\Policy}{\State{u_1}}.
\]
We examine the following cases:
\begin{itemize}
\item 
Suppose $x[u_0]=x[u_1]=0$.  Then both signal marginals are at most $1/4$, and by \eqref{eq:read-low-lambda}, both corresponding controller values are at most $19/64+45\lambda/64$.  Therefore
\[
\ContVal{0}{\State{w}}-\ContVal{1}{\State{w}}
\ge
\frac{21}{64}-\left(\frac{19}{64}+\frac{45\lambda}{64}\right)
=
\frac{2-45\lambda}{64}
>
\GateGap.
\]
Condition~\hyperref[controller-forcing]{\gadgetF} gives $\CtrlMarg{w}\le1/16$, and then condition~\hyperref[copy]{\gadgetC} gives $\RewMarg{w}<1/4$, so $x[w]=0$.

\item
Suppose at least one input is $1$, say $x[u_0]=1$.  Then $\RewMarg{u_0}\ge3/4$, so \eqref{eq:read-high-lambda} gives $\Value{\Controller}{\Policy}{\State{u_0}}\ge45/64-45\lambda/64$.  Since values are nonnegative,
\[
\ContVal{1}{\State{w}}-\ContVal{0}{\State{w}}
\ge
\frac12\left(\frac{45}{64}-\frac{45\lambda}{64}\right)-\frac{21}{64}
=
\frac{3-45\lambda}{128}
>
\GateGap.
\]
Condition~\hyperref[controller-forcing]{\gadgetF} gives $\CtrlMarg{w}\ge15/16$, and condition~\hyperref[copy]{\gadgetC} gives $\RewMarg{w}>3/4$.  Hence $x[w]=1$.
\end{itemize}

\paragraph{\PURIFY{}.}
Let $g=\PURIFY(u,v_0,v_1)$.  For the first output state, it holds that 
\[
\ContVal{0}{\State{v_0}}=\frac{89}{128},
\qquad
\ContVal{1}{\State{v_0}}=\Value{\Controller}{\Policy}{\State{u}}.
\]
We examine the following cases:

\begin{itemize}
\item 
If $\RewMarg{u}\le2/3$, then \eqref{eq:read-two-thirds-lambda} gives
\[
\ContVal{0}{\State{v_0}}-\ContVal{1}{\State{v_0}}
\ge
\frac{89}{128}-\left(\frac{88}{128}+\frac{40\lambda}{128}\right)
=
\frac{1-40\lambda}{128}
>
\GateGap.
\]
Thus by condition~\hyperref[controller-forcing]{\gadgetF} we get $\CtrlMarg{v_0}\le1/16$. Then using condition~\hyperref[copy]{\gadgetC} we get $\RewMarg{v_0}<1/4$, and therefore $x[v_0]=0$.
\item
If $\RewMarg{u}\ge3/4$, then
\[
\ContVal{1}{\State{v_0}}-\ContVal{0}{\State{v_0}}
\ge
\left(\frac{90}{128}-\frac{90\lambda}{128}\right)-\frac{89}{128}
=
\frac{1-90\lambda}{128}
>
\GateGap,
\]  
Hence by condition~\hyperref[controller-forcing]{\gadgetF} we get $\CtrlMarg{v_0}\ge15/16$. Then by condition~\hyperref[copy]{\gadgetC} we get $\RewMarg{v_0}>3/4$, and $x[v_0]=1$.
\end{itemize}
Similarly, for the second output state, it holds that
\[
\ContVal{0}{\State{v_1}}=\frac{39}{128},
\qquad
\ContVal{1}{\State{v_1}}=\Value{\Controller}{\Policy}{\State{u}}.
\]
We examine the following cases:

\begin{itemize}
\item
If $\RewMarg{u}\le1/4$, then \eqref{eq:read-low-lambda} gives
\[
\ContVal{0}{\State{v_1}}-\ContVal{1}{\State{v_1}}
\ge
\frac{39}{128}-\left(\frac{38}{128}+\frac{90\lambda}{128}\right)
=
\frac{1-90\lambda}{128}
>
\GateGap.
\]
Thus $x[v_1]=0$.  
\item
If $\RewMarg{u}\ge1/3$, then \eqref{eq:read-one-third-lambda} gives
\[
\ContVal{1}{\State{v_1}}-\ContVal{0}{\State{v_1}}
\ge
\left(\frac{40}{128}-\frac{40\lambda}{128}\right)-\frac{39}{128}
=
\frac{1-40\lambda}{128}
>
\GateGap.
\]
Thus $x[v_1]=1$.
\end{itemize}

It remains to check the \PURIFY{} semantics.  If $x[u]=0$, then $\RewMarg{u}\le1/4$, so in particular $\RewMarg{u}\le2/3$ and $\RewMarg{u}\le1/4$; the two implications above yield $x[v_0]=x[v_1]=0$.  If $x[u]=1$, then $\RewMarg{u}\ge3/4$, so in particular $\RewMarg{u}\ge3/4$ and $\RewMarg{u}\ge1/3$; hence $x[v_0]=x[v_1]=1$.  Finally, if $x[u]=\bot$, then $\RewMarg{u}\in(1/4,3/4)$.  Every real number in this interval satisfies at least one of $\RewMarg{u}\le2/3$ or $\RewMarg{u}\ge1/3$.  Therefore at least one of $x[v_0]$ and $x[v_1]$ is Boolean.  This is exactly the \PURIFY{} condition.
\end{proof}

\begin{proposition}
\label{prop:common-decoding}
Assume $0\le\lambda<1/180$.  If a stationary policy $\Policy$ of $G_\lambda(I)$ satisfies conditions~\hyperref[copy]{\gadgetC} and~\hyperref[controller-forcing]{\gadgetF} at every output state, then the decoded assignment \eqref{eq:decode} is a valid solution of the \PureCircuit{} instance $I$.
\end{proposition}

\begin{proof}
The decoded value $x[u]$ is globally well-defined because there is exactly one state $\State{u}$ and one signal marginal $\RewMarg{u}$ for each node $u$.  Consider any gate.  All its output states satisfy conditions~\hyperref[copy]{\gadgetC} and~\hyperref[controller-forcing]{\gadgetF} by assumption.  By Proposition \ref{prop:gate-soundness}, the decoded values on that gate satisfy the corresponding \NOT{}, \OR{}, or \PURIFY{} constraint.  Since this holds gate by gate, the assignment satisfies the whole instance.
\end{proof}

\section{Perfect Markov CCE hardness}
\label{sec:perfect}

In this section, we set $\lambda=0$.  The controller reward at every variable state is $a_{\RewardPlayer}$, so the controller's immediate reward is independent of its own action.

\subsection{Proving the single-controller gadget conditions}
Our goal is to prove the single-controller gadget conditions~\hyperref[copy]{\gadgetC} and~\hyperref[controller-forcing]{\gadgetF} using the fact that $\Policy$ is a
perfect Markov $\Approx$-CCE of the constructed two-player game.

\begin{lemma}[Copy gadget for perfect Markov CCE]
\label{lem:perfect-copy}
Let $\Policy$ be a perfect Markov $\Approx$-CCE of $G_0(I)$.  If $\Approx<7/128$, then for every variable state $\State{u}$, condition~\hyperref[copy]{\gadgetC} holds.
\end{lemma}

\begin{proof}
Fix $u$ and write
\[
P_{ij}:=\Pr_{\Policy}[a_{\Controller}=i,\ a_{\RewardPlayer}=j\mid \State{u}],
\qquad i,j\in\Bits.
\]
Then, using \eqref{eq:M-def} we can write the controller's marginal, the reward-player's marginal, and the matching probability as follows:
\[
\CtrlMarg{u}=P_{10}+P_{11},
\qquad
\RewMarg{u}=P_{01}+P_{11},
\qquad
\MatchProb{u}=P_{00}+P_{11}.
\]

\begin{itemize}
\item
Suppose first that $\CtrlMarg{u}\ge15/16$ but $\RewMarg{u}\le3/4$.  Then, we have
\[
P_{10}-P_{01}=\CtrlMarg{u}-\RewMarg{u}\ge \frac{15}{16}-\frac34=\frac3{16},
\]
so $P_{10}\ge3/16$.  Also $P_{00}\le P_{00}+P_{01}=1-\CtrlMarg{u}\le1/16$.  Hence
\begin{equation}
\CtrlMarg{u}-\MatchProb{u}=P_{10}-P_{00}\ge\frac3{16}-\frac1{16}=\frac18.
\nonumber
\end{equation}

Consider the stationary deviation $\Policy'_{\RewardPlayer}$ of player $\RewardPlayer$ that plays action $1$ at $\State{u}$ and arbitrary fixed actions at all other states.
At $\State{u}$, the expected immediate reward of $\RewardPlayer$ under $\Policy$ is
$\MatchProb{u}$, while the expected immediate reward under
$\Policy_{\Controller}\times\Policy'_{\RewardPlayer}$ is $\CtrlMarg{u}$.
Since $\RewardPlayer$ does not affect transitions and all continuation values lie in
$[0,1]$, the gain of this deviation from the initial state $\State{u}$ satisfies
\[
V_{\RewardPlayer}^{\Policy_{\Controller}\times\Policy'_{\RewardPlayer}}(\State{u})
-
V_{\RewardPlayer}^{\Policy}(\State{u})
\ge
(1-\Discount)\bigl(\CtrlMarg{u}-\MatchProb{u}\bigr)-\Discount
\ge
(1-\Discount)\frac18-\Discount
=
\frac7{128}.
\]
contradicting perfect Markov $\Approx$-CCE when $\Approx<7/128$.  Therefore $\RewMarg{u}>3/4$.

\item 
The other case is symmetric.  Suppose $\CtrlMarg{u}\le1/16$ but $\RewMarg{u}\ge1/4$.
Since $P_{11}\le\CtrlMarg{u}\le1/16$,
\[
P_{01}
=
\RewMarg{u}-P_{11}
\ge
\frac14-\frac1{16}
=
\frac3{16}.
\]
Also since $P_{11}\le1/16$, we have
\begin{equation}
(1-\CtrlMarg{u})-\MatchProb{u}
=
P_{01}-P_{11}
\ge
\frac3{16}-\frac1{16}
=
\frac18. 
\nonumber
\end{equation}

Consider the stationary deviation $\Policy'_{\RewardPlayer}$ of player $\RewardPlayer$
that plays action $0$ at $\State{u}$ and arbitrary fixed actions at all other states.
At $\State{u}$, the immediate reward of $\RewardPlayer$ under $\Policy$ is
$\MatchProb{u}$, while the immediate reward under
$\Policy_{\Controller}\times\Policy'_{\RewardPlayer}$ is $1-\CtrlMarg{u}$.
Since $\RewardPlayer$ does not affect transitions and all continuation values lie in
$[0,1]$, the gain of this deviation from the initial state $\State{u}$ satisfies
\[
V_{\RewardPlayer}^{\Policy_{\Controller}\times\Policy'_{\RewardPlayer}}(\State{u})
-
V_{\RewardPlayer}^{\Policy}(\State{u})
\ge
(1-\Discount)\bigl((1-\CtrlMarg{u})-\MatchProb{u}\bigr)-\Discount
\ge
(1-\Discount)\frac18-\Discount
=
\frac7{128}.
\]
This again contradicts perfect Markov $\Approx$-CCE when $\Approx<7/128$.
Hence $\RewMarg{u}<1/4$.
\end{itemize}

\end{proof}

Next, we bound the best controller's best local advantage \eqref{eq: best_advantage}, as follows:

\begin{lemma}[Best local advantage under perfect Markov CCE]
\label{lem:best-adv-perfect}
If $\Policy$ is a perfect Markov $\Approx$-CCE of $G_0(I)$, then for every variable state $s$, $\ControllerBestGap{s}\le \Approx$.
\end{lemma}

\begin{proof}
At each state $s$, we define the controller deviating action $b^\star(s)\in\argmax_b\ActionAdvantage{\Controller}{\Policy}{s}{b}$; at sink states choose the dummy action.  Let $\DeviationPolicy{\Controller}$ be the deterministic stationary controller deviation that plays $b^\star(s)$ at every state.  At every variable state, the local advantage of this deviation is exactly $\ControllerBestGap{s}$ (see \eqref{eq: best_advantage}).

When $\lambda=0$, the controller's value at a variable state is a convex combination of its two averaged action values:
\[
\Value{\Controller}{\Policy}{s}
=(1-\CtrlMarg{s})\AvgQValue{\Controller}{\Policy}{s}{0}
+\CtrlMarg{s}\AvgQValue{\Controller}{\Policy}{s}{1}.
\]
Thus $\ControllerBestGap{s}\ge0$ for every variable state, and it is also nonnegative at sinks.  Applying Lemma \ref{lem:performance-difference-deviation-policy} from initial state $s$ gives
\[
\Value{\Controller}
{\DeviationPolicy{\Controller}\times\Policy_{\RewardPlayer}}{s}-\Value{\Controller}{\Policy}{s}
=
\frac1{1-\Discount}\sum_{s'\in\StateSpace}\VisitDist{s}{\DeviationPolicy{\Controller}\times\Policy_{\RewardPlayer}}{s'}\ControllerBestGap{s'}.
\]
All summands are nonnegative, and $\VisitDist{s}{\DeviationPolicy{\Controller}\times\Policy_{\RewardPlayer}}{s}\ge1-\Discount$.  Therefore the gain is at least $\ControllerBestGap{s}$.  Perfect Markov $\Approx$-CCE forbids gain larger than $\Approx$ from state $s$, so $\ControllerBestGap{s}\le\Approx$.
\end{proof}

\begin{lemma}[Controller forcing for perfect Markov CCE]
\label{lem:perfect-force}
Let $\Policy$ be a perfect Markov $\Approx$-CCE of $G_0(I)$.  If $\Approx<\Discount\GateGap/16$, then every output state $\State{v}$ satisfies condition~\hyperref[controller-forcing]{\gadgetF}.
\end{lemma}

\begin{proof}
For $\lambda=0$, the controller's immediate reward is independent of the controller's action.  Hence, for $b\in\Bits$, we have
\begin{equation}
\AvgQValue{\Controller}{\Policy}{\State{v}}{b}
=(1-\Discount)\RewMarg{v}+\Discount\ContVal{b}{\State{v}}.
\label{eq:barQ-perfect}
\end{equation}
Also, we have 
\[
\Value{\Controller}{\Policy}{\State{v}}
=(1-\CtrlMarg{v})\AvgQValue{\Controller}{\Policy}{\State{v}}{0}
+\CtrlMarg{v}\AvgQValue{\Controller}{\Policy}{\State{v}}{1}.
\]
Suppose $\ContVal{1}{\State{v}}-\ContVal{0}{\State{v}}\ge\GateGap$.  By \eqref{eq:barQ-perfect}, we get
\[
\AvgQValue{\Controller}{\Policy}{\State{v}}{1}
-
\AvgQValue{\Controller}{\Policy}{\State{v}}{0}
=\Discount(\ContVal{1}{\State{v}}-\ContVal{0}{\State{v}})
\ge\Discount\GateGap.
\]
Therefore, we obtain
\begin{align}
\ActionAdvantage{\Controller}{\Policy}{\State{v}}{1}
&=
\AvgQValue{\Controller}{\Policy}{\State{v}}{1}
-
\Value{\Controller}{\Policy}{\State{v}}
\nonumber\\
&=
\AvgQValue{\Controller}{\Policy}{\State{v}}{1}
-
\Bigl(
\CtrlMarg{v}\AvgQValue{\Controller}{\Policy}{\State{v}}{1}
+
(1-\CtrlMarg{v})\AvgQValue{\Controller}{\Policy}{\State{v}}{0}
\Bigr)
\nonumber\\
&=
(1-\CtrlMarg{v})
\left(
\AvgQValue{\Controller}{\Policy}{\State{v}}{1}
-
\AvgQValue{\Controller}{\Policy}{\State{v}}{0}
\right)
\nonumber\\
&\ge
(1-\CtrlMarg{v})\Discount\GateGap.
\end{align}
If $\CtrlMarg{v}<15/16$, then this advantage is strictly larger than $\Discount\GateGap/16$, contradicting Lemma \ref{lem:best-adv-perfect} for $\Approx<\Discount\GateGap/16$. Thus $\CtrlMarg{v}\ge15/16$.

The other direction is identical.  If $\ContVal{0}{\State{v}}-\ContVal{1}{\State{v}}\ge\GateGap$, then the advantage of action $0$ is at least $\CtrlMarg{v}\Discount\GateGap$.  If $\CtrlMarg{v}>1/16$, this is larger than $\Discount\GateGap/16$, again contradicting Lemma \ref{lem:best-adv-perfect}.  Hence $\CtrlMarg{v}\le1/16$.
\end{proof}

\subsection{Perfect Markov CCE is \PPAD-hard}

Putting everything together, we prove the following main result:

\begin{statementbox}
\begin{theorem}[Perfect Markov $\Approx$-CCE hardness]
\label{thm:perfect}
There exists a universal constant $\varepsilon_{\mathrm 0}>0$ such that computing a perfect Markov $\varepsilon$-CCE of a discounted two-player single-controller Markov game is \PPAD-hard for every $\varepsilon<\varepsilon_{\mathrm 0}$.  The hardness holds for $\Discount=1/16$, rewards in $[0,1]$, and binary actions at all states.
\end{theorem}
\end{statementbox}

\begin{proof}
Set $\varepsilon_{\mathrm 0}:=\frac{\Discount\GateGap}{16}=\frac1{65536}$.
Given a \PureCircuit{} instance $I$, construct $G_0(I)$.  The construction has one variable state per node and constantly many sinks; all transition probabilities and rewards are rational and can be written with polynomially many bits.

Suppose a polynomial-time algorithm outputs a perfect Markov $\varepsilon$-CCE $\Policy$ for some $\varepsilon<\varepsilon_{\mathrm 0}$.  Since $\varepsilon<7/128$, Lemma \ref{lem:perfect-copy} implies condition~\hyperref[copy]{\gadgetC} at every output state.  Since $\varepsilon<\Discount\GateGap/16$, Lemma \ref{lem:perfect-force} implies condition~\hyperref[controller-forcing]{\gadgetF} at every output state.  By Proposition \ref{prop:common-decoding}, the decoded assignment solves $I$.  This gives a polynomial-time reduction from \PureCircuit{} to computing a perfect Markov $\varepsilon$-CCE.  Since \PureCircuit{} is \PPAD-hard, the theorem follows.
\end{proof}

\section{Non-perfect Markov CCE hardness}
\label{sec:non-perfect}

We now set $0<\lambda<1/180$ and again use the game $G_\lambda(I)$ defined in Section \ref{sec:construction}.  
Let $\OutputStates$ be the set of output variable states.  
We define the initial distribution $\InitialDist$ as the uniform distribution over $\OutputStates$.
Since every \NOT{} and \OR{} gate has one output and every \PURIFY{} gate has two outputs,
\begin{equation}
|\OutputStates|\le2|\Gates|.
\label{eq:output-count}
\end{equation}
We define the following constant quantities:
\begin{equation}
c_{\mathrm{copy}}:=\frac7{32},
\qquad
c_{\mathrm{force}}:=\frac{\Discount\GateGap}{16}=\frac1{65536}.
\label{eq:copy-force-constants}
\end{equation}
Let $\eta>0$ and choose $\lambda$ so that
\begin{equation}
0<\lambda<\frac1{180},
\qquad
\lambda\le\frac{c_{\mathrm{force}}}{2(1-\Discount)},
\qquad
\lambda\le\frac{\eta c_{\mathrm{force}}}{4}.
\label{eq:lambda-choice}
\end{equation}
For the constant-accuracy result under \Cref{conj:pcp-ppad}, we choose
\begin{equation}
\varepsilon'
<
\min\left\{
\frac{(1-\Discount)\eta c_{\mathrm{copy}}}{1+\lambda^{-1}},
\frac{\eta c_{\mathrm{force}}}{4}
\right\}.
\label{eq:rhoNP-choice}
\end{equation}
The parameters $\eta$, $\lambda$, and $\varepsilon'$ will be instantiated in the two hardness proofs below.

\subsection{Average copy}

For a variable state $\State{u}$, define the product-marginal matching probability
\begin{equation}
\ProdMatch{u}:=\CtrlMarg{u}\RewMarg{u}+(1-\CtrlMarg{u})(1-\RewMarg{u})
\label{eq:Mprod-def}
\end{equation}
and the best fixed matching payoff
\[
B(q):=\max\{q,1-q\}.
\]

\begin{lemma}[Decorrelation identity]
\label{lem:decorrelation-identity}
Let $\Policy$ be any stationary joint policy of $G_\lambda(I)$.  Let $\Policy_{\RewardPlayer}^{\mathrm{br}}$ be the stationary reward-player deviation that, at each variable state $\State{u}$, plays action $1$ if $\CtrlMarg{u}\ge1/2$ and plays action $0$ otherwise.  Let $\Policy_{\Controller}^{\mathrm{dec}}$ be the stationary controller deviation that, at each variable state $\State{u}$, plays action $1$ with probability $\CtrlMarg{u}$ and action $0$ with probability $1-\CtrlMarg{u}$ independently of the reward player's action.  At sinks the deviations play uniformly.  Define
\[
\Delta_{\RewardPlayer}:=\Value{\RewardPlayer}{\Policy_{\Controller}\times\Policy_{\RewardPlayer}^{\mathrm{br}}}{\InitialDist}-\Value{\RewardPlayer}{\Policy}{\InitialDist}
\quad \text{and} \quad
\Delta_{\Controller}^{\mathrm{dec}}:=\Value{\Controller}{\Policy_{\Controller}^{\mathrm{dec}}\times\Policy_{\RewardPlayer}}{\InitialDist}-\Value{\Controller}{\Policy}{\InitialDist}.
\]
Then
\begin{equation}
\Delta_{\RewardPlayer}+\lambda^{-1}\Delta_{\Controller}^{\mathrm{dec}}
=
\sum_{u\in\NodeSet}\VisitDist{\InitialDist}{\Policy}{\State{u}}
\left(B(\CtrlMarg{u})-\ProdMatch{u}\right).
\label{eq:decorrelation-identity}
\end{equation}
Moreover, every summand on the right-hand side of \eqref{eq:decorrelation-identity} is nonnegative.
\end{lemma}

\begin{proof}
The reward player's actions do not affect transitions.  Under $\Policy_{\Controller}\times\Policy_{\RewardPlayer}^{\mathrm{br}}$, the state trajectory is therefore distributed exactly as under $\Policy$.  At a variable state $\State{u}$, the reward-player payoff under $\Policy$ is $\MatchProb{u}$.  The best fixed action obtains $\max\{\CtrlMarg{u},1-\CtrlMarg{u}\}=B(\CtrlMarg{u})$.  Since the reward player receives zero at sinks under both policies, the normalized discounted value difference is
\begin{equation}
\Delta_{\RewardPlayer}
=
\sum_{u\in\NodeSet}\VisitDist{\InitialDist}{\Policy}{\State{u}}
\left(B(\CtrlMarg{u})-\MatchProb{u}\right).
\label{eq:DeltaR-copy}
\end{equation}

Now consider the controller decorrelation deviation.  At every variable state $\State{u}$, the controller marginal under the deviation is still $\CtrlMarg{u}$.  Because transitions depend only on the controller's current action, the transition kernel induced by the deviated profile equals the transition kernel induced by $\Policy$ at every state.  Thus the normalized visitation distribution is again $\VisitDist{\InitialDist}{\Policy}{\cdot}$.

At $\State{u}$, the payoff $(1-\lambda)a_{\RewardPlayer}$ has the same expectation under the deviation and under $\Policy$, because the reward player's marginal remains $\RewMarg{u}$.  The only changed term is the mismatch bonus.  Under $\Policy$, the mismatch probability is $1-\MatchProb{u}$.  Under the independent product of the same marginals, the match probability is $\ProdMatch{u}$, so the mismatch probability is $1-\ProdMatch{u}$.  Hence the expected controller reward at $\State{u}$ changes by
\[
\lambda\bigl((1-\ProdMatch{u})-(1-\MatchProb{u})\bigr)
=\lambda(\MatchProb{u}-\ProdMatch{u}).
\]
Therefore
\begin{equation}
\Delta_{\Controller}^{\mathrm{dec}}
=
\lambda\sum_{u\in\NodeSet}\VisitDist{\InitialDist}{\Policy}{\State{u}}
\left(\MatchProb{u}-\ProdMatch{u}\right).
\label{eq:DeltaC-dec}
\end{equation}
Adding \eqref{eq:DeltaR-copy} and $\lambda^{-1}$ times \eqref{eq:DeltaC-dec} proves \eqref{eq:decorrelation-identity}.

It remains to prove nonnegativity.  If $q\ge1/2$, then
\[
B(q)-\bigl(qp+(1-q)(1-p)\bigr)
=q-qp-(1-q)(1-p)=(1-p)(2q-1)\ge0.
\]
If $q\le1/2$, then
\[
B(q)-\bigl(qp+(1-q)(1-p)\bigr)
=(1-q)-qp-(1-q)(1-p)=p(1-2q)\ge0.
\]
This proves the claim state by state.
\end{proof}

Define the set of copy failures among output states by
\[
\Copy:=\left\{\State{u}\in\OutputStates:\
\begin{array}{l}
\CtrlMarg{u}\ge15/16\text{ and }\RewMarg{u}\le3/4,\ \text{or}\\
\CtrlMarg{u}\le1/16\text{ and }\RewMarg{u}\ge1/4
\end{array}
\right\}.
\]

\begin{lemma}[Average copy lemma]
\label{lem:average-copy}
Let $\Policy$ be a non-perfect Markov $\Approx$-CCE of $G_\lambda(I)$ from the uniform distribution on $\OutputStates$. Then, the following holds:
\[
\text{If} \quad
\Approx<\frac{(1-\Discount)\eta c_{\mathrm{copy}}}{1+\lambda^{-1}}, \quad \text{then} \quad \InitialDist(\Copy)<\eta,
\]
where $\InitialDist(\Copy):= \sum_{s\in\Copy}\InitialDist(s)$.
\end{lemma}

\begin{proof}
First suppose that $q:=\CtrlMarg{u}\ge15/16$ and $p:=\RewMarg{u}\le3/4$.  Since $q\ge1/2$,
\[
B(q)-\bigl(qp+(1-q)(1-p)\bigr)=(1-p)(2q-1)
\ge\frac14\cdot\frac78=\frac7{32}=c_{\mathrm{copy}}.
\]
Now suppose that $q\le1/16$ and $p\ge1/4$.  Since $q\le1/2$,
\[
B(q)-\bigl(qp+(1-q)(1-p)\bigr)=p(1-2q)
\ge\frac14\cdot\frac78=\frac7{32}=c_{\mathrm{copy}}.
\]
Thus every state in $\Copy$ contributes at least $c_{\mathrm{copy}}$ to the summand in Lemma \ref{lem:decorrelation-identity}.

Assume toward a contradiction that $\InitialDist(\Copy)\ge\eta$.  Using \eqref{eq:visit-lower}, the identity \eqref{eq:decorrelation-identity}, and nonnegativity of all summands,
\[
\Delta_{\RewardPlayer}+\lambda^{-1}\Delta_{\Controller}^{\mathrm{dec}}
\ge
\sum_{\State{u}\in\Copy}\VisitDist{\InitialDist}{\Policy}{\State{u}}c_{\mathrm{copy}}
\ge
(1-\Discount)\InitialDist(\Copy)c_{\mathrm{copy}}
\ge
(1-\Discount)\eta c_{\mathrm{copy}}.
\]
On the other hand, non-perfect Markov $\Approx$-CCE implies $\Delta_{\RewardPlayer}\le\Approx$ and $\Delta_{\Controller}^{\mathrm{dec}}\le\Approx$, because both are gains from valid stationary one-player deviations.  Hence the left-hand side of \eqref{eq:decorrelation-identity} is at most $\Approx(1+
\lambda^{-1})$, contradicting the assumed bound on $\Approx$.  Therefore $\InitialDist(\Copy)<\eta$.
\end{proof}

\subsection{Average controller-forcing}

Define the set of forcing failures among output states by
\[
\Force:=\left\{\State{v}\in\OutputStates:\
\begin{array}{l}
\ContVal{1}{\State{v}}-\ContVal{0}{\State{v}}\ge\GateGap\text{ and }\CtrlMarg{v}<15/16,\ \text{or}\\
\ContVal{0}{\State{v}}-\ContVal{1}{\State{v}}\ge\GateGap\text{ and }\CtrlMarg{v}>1/16
\end{array}
\right\}.
\]

\begin{lemma}[Average forcing lemma]
\label{lem:average-forcing}
Let $\Policy$ be a non-perfect Markov $\Approx$-CCE of $G_\lambda(I)$ from the uniform distribution on $\OutputStates$, and suppose \eqref{eq:lambda-choice} holds.  If $\Approx<\eta c_{\mathrm{force}}/4$, then $\InitialDist(\Force)<\eta$.
\end{lemma}

\begin{proof}
Assume toward a contradiction that
\[
\InitialDist(\Force)\ge\eta.
\]
Define the stationary controller deviation
$\Policy_{\Controller}^{F}$ at every variable state by
\begin{equation}
\Policy_{\Controller}^{F}(1\mid\State{u})
:=
\begin{cases}
1,
&
\State{u}\in\Force
\text{ and }
\ContVal{1}{\State{u}}-\ContVal{0}{\State{u}}
\ge\GateGap,
\\[1mm]
0,
&
\State{u}\in\Force
\text{ and }
\ContVal{0}{\State{u}}-\ContVal{1}{\State{u}}
\ge\GateGap,
\\[1mm]
\CtrlMarg{u},
&
\State{u}\notin\Force.
\end{cases}
\label{eq:forcing-deviation-definition}
\end{equation}
At every sink state, $\Policy_{\Controller}^{F}$ plays uniformly.  Let
$\widetilde{\Policy}
:=
\Policy_{\Controller}^{F}\times\Policy_{\RewardPlayer}$
denote the resulting deviated joint policy. 

\noindent
By
\eqref{eq:lambda-reward}, \eqref{eq:pq-def},
\eqref{eq:M-def}, and \eqref{eq:W-def}, we have
\begin{align}
\Value{\Controller}{\Policy}{\State{u}}
&=
\sum_{\Action{\Controller}\in\Bits}
\sum_{\Action{\RewardPlayer}\in\Bits}
\Policy(
\Action{\Controller},
\Action{\RewardPlayer}
\mid\State{u})
\QValue{\Controller}{\Policy}
{\State{u},
(\Action{\Controller},\Action{\RewardPlayer})}
\nonumber\\
&=
(1-\Discount)
\sum_{\Action{\Controller}\in\Bits}
\sum_{\Action{\RewardPlayer}\in\Bits}
\Policy(
\Action{\Controller},
\Action{\RewardPlayer}
\mid\State{u})
\Bigl[
(1-\lambda)\Action{\RewardPlayer}
+
\lambda
\ind[
\Action{\Controller}\ne\Action{\RewardPlayer}
]
\Bigr]
\nonumber\\
&\qquad
+
\Discount
\sum_{\Action{\Controller}\in\Bits}
\sum_{\Action{\RewardPlayer}\in\Bits}
\Policy(
\Action{\Controller},
\Action{\RewardPlayer}
\mid\State{u})
\ContVal{\Action{\Controller}}{\State{u}}
\nonumber\\
&=
(1-\Discount)
\Bigl[
(1-\lambda)\RewMarg{u}
+
\lambda(1-\MatchProb{u})
\Bigr]
\nonumber\\
&\qquad
+
\Discount
\Bigl[
(1-\CtrlMarg{u})\ContVal{0}{\State{u}}
+
\CtrlMarg{u}\ContVal{1}{\State{u}}
\Bigr].
\label{eq:forcing-original-value}
\end{align}
\noindent\textbf{Forcing failure with action 1 preferred.} Consider first a state $\State{v}\in\Force$ satisfying
\begin{equation}
\ContVal{1}{\State{v}}-\ContVal{0}{\State{v}}
\ge\GateGap,
\qquad
\CtrlMarg{v}<\HighCtrlThreshold.
\label{eq:forcing-failure-one-case}
\end{equation}
By \eqref{eq:forcing-deviation-definition}, the deviation plays
action $1$ deterministically at $\State{v}$.  Using
\eqref{eq:barQ-def} and \eqref{eq:lambda-reward},
\begin{align}
\AvgQValue{\Controller}{\Policy}{\State{v}}{1}
&=
\sum_{\Action{\RewardPlayer}\in\Bits}
\Policy_{\RewardPlayer}
(\Action{\RewardPlayer}\mid\State{v})
\QValue{\Controller}{\Policy}
{\State{v},(1,\Action{\RewardPlayer})}
\nonumber\\
&=
\sum_{\Action{\RewardPlayer}\in\Bits}
\Policy_{\RewardPlayer}
(\Action{\RewardPlayer}\mid\State{v})
\Bigl[
(1-\Discount)
\bigl(
(1-\lambda)\Action{\RewardPlayer}
+
\lambda\ind[1\ne\Action{\RewardPlayer}]
\bigr)
+
\Discount\ContVal{1}{\State{v}}
\Bigr]
\nonumber\\
&=
(1-\Discount)
\Bigl[
(1-\lambda)\RewMarg{v}
+
\lambda(1-\RewMarg{v})
\Bigr]
+
\Discount\ContVal{1}{\State{v}}.
\label{eq:forcing-barQ-one}
\end{align}
Consequently, by \eqref{eq:forcing-original-value}, we get
\begin{align}
\CoarseAdvantage{\Controller}{\Policy}
{\Policy_{\Controller}^{F}}{\State{v}}
&=
\ActionAdvantage{\Controller}{\Policy}{\State{v}}{1}
\nonumber\\
&=
\AvgQValue{\Controller}{\Policy}{\State{v}}{1}
-
\Value{\Controller}{\Policy}{\State{v}}
\nonumber\\
&=
(1-\Discount)
\Bigl[
(1-\lambda)\RewMarg{v}
+
\lambda(1-\RewMarg{v})
\Bigr]
+
\Discount\ContVal{1}{\State{v}}
\nonumber\\
&\qquad
-
(1-\Discount)
\Bigl[
(1-\lambda)\RewMarg{v}
+
\lambda(1-\MatchProb{v})
\Bigr]
\nonumber\\
&\qquad
-
\Discount
\Bigl[
(1-\CtrlMarg{v})\ContVal{0}{\State{v}}
+
\CtrlMarg{v}\ContVal{1}{\State{v}}
\Bigr]
\nonumber\\
&=
(1-\Discount)\lambda
\Bigl[
(1-\RewMarg{v})-(1-\MatchProb{v})
\Bigr]
\nonumber\\
&\qquad
+
\Discount
\Bigl[
\ContVal{1}{\State{v}}
-
(1-\CtrlMarg{v})\ContVal{0}{\State{v}}
-
\CtrlMarg{v}\ContVal{1}{\State{v}}
\Bigr]
\nonumber\\
&=
(1-\Discount)\lambda
\bigl(
\MatchProb{v}-\RewMarg{v}
\bigr)
\nonumber\\
&\qquad
+
\Discount(1-\CtrlMarg{v})
\bigl(
\ContVal{1}{\State{v}}
-
\ContVal{0}{\State{v}}
\bigr).
\label{eq:forcing-advantage-one}
\end{align}
Because
$1-\CtrlMarg{v}>\frac1{16}$ and
$\MatchProb{v}-\RewMarg{v}\ge-1$, 
\eqref{eq:forcing-advantage-one} gives
\begin{align}
\CoarseAdvantage{\Controller}{\Policy}
{\Policy_{\Controller}^{F}}{\State{v}}
&\ge
\Discount(1-\CtrlMarg{v})\GateGap
-
(1-\Discount)\lambda
\nonumber\\
&>
\frac{\Discount\GateGap}{16}
-
(1-\Discount)\lambda
\nonumber\\
&=
c_{\mathrm{force}}
-
(1-\Discount)\lambda,
\label{eq:forcing-advantage-one-lower}
\end{align}
where the last equality follows from
\eqref{eq:copy-force-constants}.

\noindent\textbf{Forcing failure with action 0 preferred.}
Now consider a state $\State{v}\in\Force$ satisfying
\begin{equation}
\ContVal{0}{\State{v}}-\ContVal{1}{\State{v}}
\ge\GateGap
\qquad \text{and} \qquad
\CtrlMarg{v}>\LowCtrlThreshold.
\label{eq:forcing-failure-zero-case}
\end{equation}
The deviation plays action $0$ deterministically at $\State{v}$.
Its averaged action value is
\begin{align}
\AvgQValue{\Controller}{\Policy}{\State{v}}{0}
&=
\sum_{\Action{\RewardPlayer}\in\Bits}
\Policy_{\RewardPlayer}
(\Action{\RewardPlayer}\mid\State{v})
\QValue{\Controller}{\Policy}
{\State{v},(0,\Action{\RewardPlayer})}
\nonumber\\
&=
\sum_{\Action{\RewardPlayer}\in\Bits}
\Policy_{\RewardPlayer}
(\Action{\RewardPlayer}\mid\State{v})
\Bigl[
(1-\Discount)
\bigl(
(1-\lambda)\Action{\RewardPlayer}
+
\lambda\ind[0\ne\Action{\RewardPlayer}]
\bigr)
+
\Discount\ContVal{0}{\State{v}}
\Bigr]
\nonumber\\
&=
(1-\Discount)
\Bigl[
(1-\lambda)\RewMarg{v}
+
\lambda\RewMarg{v}
\Bigr]
+
\Discount\ContVal{0}{\State{v}}.
\label{eq:forcing-barQ-zero}
\end{align}
Thus, we obtain
\begin{align}
\CoarseAdvantage{\Controller}{\Policy}
{\Policy_{\Controller}^{F}}{\State{v}}
&=
\ActionAdvantage{\Controller}{\Policy}{\State{v}}{0}
\nonumber\\
&=
\AvgQValue{\Controller}{\Policy}{\State{v}}{0}
-
\Value{\Controller}{\Policy}{\State{v}}
\nonumber\\
&=
(1-\Discount)\lambda
\Bigl[
\RewMarg{v}-(1-\MatchProb{v})
\Bigr]
\nonumber\\
&\qquad
+
\Discount
\Bigl[
\ContVal{0}{\State{v}}
-
(1-\CtrlMarg{v})\ContVal{0}{\State{v}}
-
\CtrlMarg{v}\ContVal{1}{\State{v}}
\Bigr]
\nonumber\\
&=
(1-\Discount)\lambda
\bigl(
\MatchProb{v}+\RewMarg{v}-1
\bigr)
\nonumber\\
&\qquad
+
\Discount\CtrlMarg{v}
\bigl(
\ContVal{0}{\State{v}}
-
\ContVal{1}{\State{v}}
\bigr).
\label{eq:forcing-advantage-zero}
\end{align}
Since
$\CtrlMarg{v}>\LowCtrlThreshold$ and
$\MatchProb{v}+\RewMarg{v}-1\ge-1$,
we obtain
\begin{align}
\CoarseAdvantage{\Controller}{\Policy}
{\Policy_{\Controller}^{F}}{\State{v}}
&\ge
\Discount\CtrlMarg{v}\GateGap
-
(1-\Discount)\lambda
\nonumber\\
&>
\frac{\Discount\GateGap}{16}
-
(1-\Discount)\lambda
\nonumber\\
&=
c_{\mathrm{force}}
-
(1-\Discount)\lambda.
\label{eq:forcing-advantage-zero-lower}
\end{align}
Combining \eqref{eq:forcing-advantage-one-lower} and
\eqref{eq:forcing-advantage-zero-lower}, every forcing-failure
state satisfies
\begin{equation}
\CoarseAdvantage{\Controller}{\Policy}
{\Policy_{\Controller}^{F}}{\State{v}}
\ge
c_{\mathrm{force}}
-
(1-\Discount)\lambda
\qquad
\text{for every }\State{v}\in\Force.
\label{eq:forcing-state-lower}
\end{equation}

\noindent\textbf{States outside $\Force$.} Next, let $\State{u}$ be a variable state outside $\Force$.
By \eqref{eq:forcing-deviation-definition}, the deviation plays
action $1$ with probability $\CtrlMarg{u}$ and action $0$ with
probability $1-\CtrlMarg{u}$, independently of the reward player's
action.  Therefore the controller marginal is unchanged.  The
expected continuation term is consequently
\[
(1-\CtrlMarg{u})\ContVal{0}{\State{u}}
+
\CtrlMarg{u}\ContVal{1}{\State{u}},
\]
which is the same continuation term as in
\eqref{eq:forcing-original-value}.

The reward-player marginal also remains $\RewMarg{u}$.  Under the
independent product of the two marginals, the match probability is
$\ProdMatch{u}$, as defined in \eqref{eq:Mprod-def}.  Hence
\begin{align}
&
\sum_{b_{\Controller}\in\Bits}
\sum_{\Action{\RewardPlayer}\in\Bits}
\widetilde{\Policy}
(b_{\Controller},\Action{\RewardPlayer}\mid\State{u})
\QValue{\Controller}{\Policy}
{\State{u},(b_{\Controller},\Action{\RewardPlayer})}
\nonumber\\
&\qquad=
(1-\Discount)
\Bigl[
(1-\lambda)\RewMarg{u}
+
\lambda(1-\ProdMatch{u})
\Bigr]
\nonumber\\
&\qquad\qquad
+
\Discount
\Bigl[
(1-\CtrlMarg{u})\ContVal{0}{\State{u}}
+
\CtrlMarg{u}\ContVal{1}{\State{u}}
\Bigr].
\label{eq:forcing-nonfailure-deviated-q}
\end{align}
Subtracting \eqref{eq:forcing-original-value} from
\eqref{eq:forcing-nonfailure-deviated-q} gives
\begin{align}
\CoarseAdvantage{\Controller}{\Policy}
{\Policy_{\Controller}^{F}}{\State{u}}
&=
(1-\Discount)\lambda
\Bigl[
(1-\ProdMatch{u})-(1-\MatchProb{u})
\Bigr]
\nonumber\\
&=
(1-\Discount)\lambda
\bigl(
\MatchProb{u}-\ProdMatch{u}
\bigr)
\nonumber\\
&\ge
-(1-\Discount)\lambda,
\qquad
\State{u}\notin\Force,
\label{eq:forcing-nonfailure-lower}
\end{align}
because $\MatchProb{u},\ProdMatch{u}\in[0,1]$.

At every sink state $z$, the controller has a unique dummy action and the state self-loops; hence the deviation cannot change either the reward or the transition, so
\begin{equation}
\CoarseAdvantage{\Controller}{\Policy}
{\Policy_{\Controller}^{F}}{z}
=0.
\label{eq:forcing-sink-zero}
\end{equation}
We now apply the performance-difference lemma (Lemma \ref{lem:performance-difference-deviation-policy}).  Mixing
\eqref{eq:coarse-perf-diff}, over the initial state
$s_0\sim\InitialDist$ gives
\begin{align}
\Value{\Controller}{\widetilde{\Policy}}{\InitialDist}
-
\Value{\Controller}{\Policy}{\InitialDist}
&=
\sum_{s_0\in\StateSpace}
\InitialDist(s_0)
\left(
\Value{\Controller}{\widetilde{\Policy}}{s_0}
-
\Value{\Controller}{\Policy}{s_0}
\right)
\nonumber\\
&=
\frac{1}{1-\Discount}
\sum_{s_0\in\StateSpace}
\InitialDist(s_0)
\sum_{s\in\StateSpace}
\VisitDist{s_0}{\widetilde{\Policy}}{s}
\CoarseAdvantage{\Controller}{\Policy}
{\Policy_{\Controller}^{F}}{s}
\nonumber\\
&=
\frac{1}{1-\Discount}
\sum_{s\in\StateSpace}
\VisitDist{\InitialDist}{\widetilde{\Policy}}{s}
\CoarseAdvantage{\Controller}{\Policy}
{\Policy_{\Controller}^{F}}{s}.
\label{eq:forcing-pdl-mu}
\end{align}
By \eqref{eq:lambda-choice}, we obtain
\begin{align}
c_{\mathrm{force}}-(1-\Discount)\lambda
&\ge
c_{\mathrm{force}}
-
(1-\Discount)
\frac{c_{\mathrm{force}}}{2(1-\Discount)}
\nonumber\\
&=
\frac{c_{\mathrm{force}}}{2}
>0.
\label{eq:forcing-positive-coefficient}
\end{align}
Thus, using \eqref{eq:forcing-state-lower},
\eqref{eq:forcing-nonfailure-lower}, and
\eqref{eq:forcing-sink-zero} in \eqref{eq:forcing-pdl-mu}, we get
\begin{align}
\Value{\Controller}{\widetilde{\Policy}}{\InitialDist}
-
\Value{\Controller}{\Policy}{\InitialDist}
&\ge
\frac{1}{1-\Discount}
\Biggl[
\sum_{\State{v}\in\Force}
\VisitDist{\InitialDist}{\widetilde{\Policy}}{\State{v}}
\Bigl(
c_{\mathrm{force}}
-
(1-\Discount)\lambda
\Bigr)
-
\sum_{\substack{\State{u}\text{ variable}\\
                  \State{u}\notin\Force}}
\VisitDist{\InitialDist}{\widetilde{\Policy}}{\State{u}}
(1-\Discount)\lambda
\Biggr].
\label{eq:forcing-partitioned-gain}
\end{align}
Applying \eqref{eq:visit-lower} to the stationary policy
$\widetilde{\Policy}$ gives
\begin{align}
\sum_{\State{v}\in\Force}
\VisitDist{\InitialDist}{\widetilde{\Policy}}{\State{v}}
&\ge
(1-\Discount)
\sum_{\State{v}\in\Force}
\InitialDist(\State{v})
\nonumber\\
&=
(1-\Discount)\InitialDist(\Force).
\label{eq:forcing-force-mass}
\end{align}
Moreover, it holds that
\begin{equation}
\sum_{\substack{\State{u}\text{ variable}\\
                \State{u}\notin\Force}}
\VisitDist{\InitialDist}{\widetilde{\Policy}}{\State{u}}
\le1.
\label{eq:forcing-negative-mass}
\end{equation}
Substituting \eqref{eq:forcing-force-mass} and
\eqref{eq:forcing-negative-mass} into
\eqref{eq:forcing-partitioned-gain}, and using the positivity in
\eqref{eq:forcing-positive-coefficient}, yields
\begin{align}
\Value{\Controller}{\widetilde{\Policy}}{\InitialDist}
-
\Value{\Controller}{\Policy}{\InitialDist}
&\ge
\frac{1}{1-\Discount}
\Bigl[
(1-\Discount)\InitialDist(\Force)
\bigl(
c_{\mathrm{force}}
-
(1-\Discount)\lambda
\bigr)
-
(1-\Discount)\lambda
\Bigr]
\nonumber\\
&=
\InitialDist(\Force)
\bigl(
c_{\mathrm{force}}
-
(1-\Discount)\lambda
\bigr)
-
\lambda
\nonumber\\
&\ge
\eta
\bigl(
c_{\mathrm{force}}
-
(1-\Discount)\lambda
\bigr)
-
\lambda
\nonumber\\
&\ge
\eta
\left(
c_{\mathrm{force}}
-
\frac{c_{\mathrm{force}}}{2}
\right)
-
\frac{\eta c_{\mathrm{force}}}{4}
\nonumber\\
&=
\frac{\eta c_{\mathrm{force}}}{4}.
\label{eq:forcing-global-gain}
\end{align}
Since $\Approx<\frac{\eta c_{\mathrm{force}}}{4}$, \eqref{eq:forcing-global-gain} exhibits a stationary
one-player controller deviation whose gain is strictly greater than
$\Approx$.  This contradicts the assumption that $\Policy$ is a
non-perfect Markov $\Approx$-CCE from $\InitialDist$.  Therefore,
$\InitialDist(\Force)<\eta$, which concludes the proof.
\end{proof}

\subsection{Non-Perfect Markov CCE is \PPAD-hard}

\begin{proposition}
\label{prop:bad-gates}
Let $\Policy$ be any stationary policy of $G_\lambda(I)$ with $0<\lambda<1/180$.  Let $\Bad:=\Copy\cup\Force$.  If no output state of a gate lies in $\Bad$, then the decoded assignment satisfies that gate.  Consequently, if $\InitialDist(\Bad)<\delta/2$, then the decoded assignment satisfies at least a $(1-\delta)$ fraction of the gates.
\end{proposition}

\begin{proof}
If no output state of a gate lies in $\Bad$, then neither the copy implication nor the controller-forcing implication fails at any output of that gate.  
Hence conditions~\hyperref[copy]{\gadgetC} and~\hyperref[controller-forcing]{\gadgetF} hold at all output states of the gate.  By Proposition \ref{prop:gate-soundness}, the decoded assignment satisfies the gate.

In the case that there exist unsatisfied gates, then every unsatisfied gate has at least one output state in $\Bad$.  Since each output state belongs to at most one gate, the number of unsatisfied gates is at most $|\Bad|$.  The distribution $\InitialDist$ is uniform on $\OutputStates$, so $|\Bad|=\InitialDist(\Bad)|\OutputStates|$.  Using \eqref{eq:output-count}, if $\InitialDist(\Bad)<\delta/2$ then
\[
\#\{\text{unsatisfied gates}\}
\le
|\Bad|
=
\InitialDist(\Bad)|\OutputStates|
<
\frac\delta2\cdot 2|\Gates|
=
\delta|\Gates|.
\]
Thus at least a $(1-\delta)$ fraction of gates are satisfied.
\end{proof}
Putting everything together, we prove our second main result.
\begin{statementbox}
\begin{theorem}[Non-perfect Markov $\Approx$-CCE hardness under \PCP-for-\PPAD]
\label{thm:nonperfect}
Assume \Cref{conj:pcp-ppad}, and let
$\delta_{\mathrm{pcp}}>0$ be the constant guaranteed by the conjecture.
Then there exists a constant
$\varepsilon'
=\varepsilon'(\delta_{\mathrm{pcp}})>0$
such that computing a non-perfect Markov $\varepsilon$-CCE of a discounted two-player single-controller Markov game is \PPAD-hard for every $\varepsilon<\varepsilon'$.  The hardness holds with fixed discount factor $\Discount=1/16$, rewards in $[0,1]$, and binary actions at all states.
\end{theorem}
\end{statementbox}

 \begin{proof}
Let $\delta=\delta_{\mathrm{pcp}}$, let $\eta=\delta/4$, choose $\lambda$ satisfying \eqref{eq:lambda-choice}, and choose $\varepsilon'$ satisfying \eqref{eq:rhoNP-choice}.  Given an instance $I$ of the gap \PureCircuit{} problem from \Cref{conj:pcp-ppad}, construct $G_\lambda(I)$ and let $\InitialDist$ be uniform over $\OutputStates$.

Suppose a polynomial-time algorithm outputs a non-perfect Markov $\varepsilon$-CCE $\Policy$ for some $\varepsilon<\varepsilon'$.  By Lemma \ref{lem:average-copy}, $\InitialDist(\Copy)<\eta$.  By Lemma \ref{lem:average-forcing}, $\InitialDist(\Force)<\eta$.  Therefore
\[
\InitialDist(\Bad)
\le
\InitialDist(\Copy)+\InitialDist(\Force)
<2\eta
=\delta/2.
\]
By Proposition \ref{prop:bad-gates}, the assignment decoded from $\Policy$ satisfies at least a $(1-\delta)$ fraction of the gates of $I$.  This solves $\delta$-\PureCircuit{}, which is \PPAD-hard under \Cref{conj:pcp-ppad}.  Hence computing a constant-approximate non-perfect Markov CCE in two-player single-controller Markov games is \PPAD-hard under the conjecture.
\end{proof}

The same analysis also gives an unconditional inverse-polynomial-accuracy result.
\begin{statementbox}
\begin{theorem}[Unconditional non-perfect Markov $\Approx$-CCE hardness]
\label{thm:nonperfect-unconditional}
There exists a universal constant $\kappa>0$ such that computing a non-perfect Markov $\varepsilon$-CCE of a discounted two-player single-controller Markov game is \PPAD-hard whenever $\varepsilon<\frac{\kappa}{|\StateSpace|^2}$.
The hardness holds with fixed discount factor $\Discount=1/16$, rewards in $[0,1]$, and binary actions at all states.
\end{theorem}
\end{statementbox}

\begin{proof}
Given a \PureCircuit{} instance $I$, construct $G_\lambda(I)$ as above.
Also let
\begin{equation}
\eta:=\frac{1}{|\OutputStates|},
\qquad
\lambda:=\frac{c_{\mathrm{force}}}{8|\OutputStates|}.
\nonumber
\label{eq:unconditional-parameters}
\end{equation}
The number $\lambda$ is a positive rational with polynomially many bits.  Moreover, we have
\[
\lambda<\frac1{180},
\qquad
\lambda\le\frac{c_{\mathrm{force}}}{2(1-\Discount)},
\qquad
\lambda=\frac{\eta c_{\mathrm{force}}}{8}
\le\frac{\eta c_{\mathrm{force}}}{4},
\]
so \eqref{eq:lambda-choice} holds.
Suppose a polynomial-time algorithm outputs a non-perfect Markov $\varepsilon$-CCE $\Policy$ with
\[
\varepsilon<\frac{\kappa}{|\StateSpace|^2}, \qquad \text{where} \qquad \kappa
:=
\frac{(1-\Discount)c_{\mathrm{copy}}c_{\mathrm{force}}}{32}. \nonumber
\]
Since $\lambda<1$, we have
\begin{align}
\frac{(1-\Discount)\eta c_{\mathrm{copy}}}{1+\lambda^{-1}}
&=
\frac{(1-\Discount)c_{\mathrm{copy}}}{|\OutputStates|}
\frac{\lambda}{1+\lambda}
\nonumber\\
&>
\frac{(1-\Discount)c_{\mathrm{copy}}}{|\OutputStates|}
\frac{\lambda}{2}
\nonumber\\
&=
\frac{(1-\Discount)c_{\mathrm{copy}}c_{\mathrm{force}}}{16|\OutputStates|^2}
\nonumber\\
&=
\frac{2\kappa}{|\OutputStates|^2}
\ge
\frac{2\kappa}{|\StateSpace|^2}
>
\varepsilon.
\nonumber
\label{eq:unconditional-copy-threshold}
\end{align}
Thus Lemma \ref{lem:average-copy} gives
\[
\InitialDist(\Copy)<\frac{1}{|\OutputStates|}.
\]
Also, since $\kappa<c_{\mathrm{force}}/4$ and $|\OutputStates|\le |\StateSpace|$,
\begin{align}
\frac{\eta c_{\mathrm{force}}}{4}
&=
\frac{c_{\mathrm{force}}}{4|\OutputStates|}
>
\frac{\kappa}{|\OutputStates|}
\ge
\frac{\kappa}{|\StateSpace|}
\ge
\frac{\kappa}{|\StateSpace|^2}
>
\varepsilon.
\nonumber
\label{eq:unconditional-force-threshold}
\end{align}
Hence Lemma \ref{lem:average-forcing} gives
\[
\InitialDist(\Force)<\frac{1}{|\OutputStates|}.
\]

Because $\InitialDist$ is uniform on the $|\OutputStates|$ output states, every nonempty subset of $\OutputStates$ has mass at least $1/|\OutputStates|$.  Therefore, we have
\[
\Copy=\emptyset
\qquad \text{and} \qquad
\Force=\emptyset.
\]
Conditions~\hyperref[copy]{\gadgetC} and~\hyperref[controller-forcing]{\gadgetF} hold at every output state, so Proposition \ref{prop:common-decoding} implies that the assignment decoded from $\Policy$ is a valid solution of $I$.  This gives a polynomial-time reduction from \PureCircuit{} to computing a non-perfect Markov $\varepsilon$-CCE with $\varepsilon<\kappa/|\StateSpace|^2$.  Since \PureCircuit{} is \PPAD-hard, the theorem follows.
\end{proof}

\begin{remark}\label{rem:discount}
    \textit{In a similar spirit as \cite{daskalakis2023complexity}, our results both for the perfect and non-perfect notion continue to hold for any $\gamma \in [1/16,1)$. This is due to the simple observation that by adding an action-independent lazy self-loop in the transition model in the form of $\widetilde P(s'\mid s,a) = (1-\tau)\mathbbm{1}\{s'=s\}+\tau P(s'\mid s,a)$, where $\tau =\frac{\beta(1-\gamma)}{\gamma(1-\beta)}$, and $\beta=\frac1{16}$, the construction preserves all stationary Markov CCE incentive inequalities exactly, and therefore all hardness results continue to hold for every fixed discount factor $\gamma\in[1/16,1)$.} 
\end{remark}

\input{example}

\input{conclusion}

\subsection*{Acknowledgements}
Ioannis Panageas is supported by NSF grant CCF-2454115. This research was also supported in part by project MIS 5154714 of the National Recovery and Resilience Plan Greece 2.0 funded by the European Union under the NextGenerationEU Program. GF is supported in part by NSF Award CCF-2443068, ONR grant N00014-25-1-2296, and an AI2050 Early Career Fellowship.

\bibliographystyle{alpha}  
\bibliography{main}

\newpage

\input{appendix}

\end{document}

%% file: abstract.tex
We study the complexity of computing stationary Markov coarse correlated equilibria (CCE) in discounted \textit{single-controller stochastic (Markov) games} \cite{ParthasarathyRaghavan1981,FilarVrieze2012}, a fundamental subclass of stochastic games in which all players may affect rewards, but only one player controls the state transitions.
Prior work \cite{daskalakis2023complexity, sidford_complexity, HansenNie2025} established \PPAD-hardness for computing stationary Markov CCE in two-player general-sum stochastic games via turn-based constructions in which each state is controlled by a single player, with control alternating across states. This structure forces every Markov CCE to collapse to a Nash equilibrium (NE), so hardness for NE transfers immediately to CCE. It remained open whether hardness persists when a single player controls all transitions\textemdash a setting where no such collapse occurs. 
 
We resolve this question: computing an approximate stationary Markov CCE in two-player single-controller stochastic games is \PPAD-complete, even with a fixed discount factor and binary actions. For the perfect notion (equilibrium constraints at every state) this holds unconditionally at constant accuracy; for the non-perfect notion, we prove constant-accuracy hardness under the \PCP-for-\PPAD hypothesis \cite{babichenko2016can, deligkas2026fisher} and inverse-polynomial-accuracy hardness unconditionally. To the best of our knowledge, our result is the first to show hardness for computing CCE without relying on equilibrium collapse phenomena or other routes through Nash-like structure \cite{golowich-regret, anagnostides2023complexity, peng2024complexity}. Instead, we construct single-controller gadgets whose local incentive constraints force a solution of a \textsc{Pure-Circuit} instance even under strongly correlated stationary policies.


%% file: introduction.tex
\section{Introduction}

Since the seminal work of Shapley~\cite{Shapley1953}, Markov games, also known as stochastic games, have served as the multi-agent generalization of Markov decision processes. At each state of a Markov game, all players simultaneously choose actions, receive rewards, and induce a transition to the next state. These games provide a natural language for multi-agent reinforcement learning, dynamic competition, control, and economics~\cite{Littman1994,Basar1986,BasarOlsder1998}. 


To define an equilibrium notion in a Markov game, it is crucial to define (i) how strategies for the game are encoded; and (ii) what deviations from a putative equilibrium one considers.
A policy is \textit{Markovian} if its prescription depends only on the current state, rather than on the full history of play. 
Policies can also be \textit{nonstationary}, allowing this prescription to vary with time, or \textit{stationary}, using the same prescription whenever the same state is visited. 
For finite discounted stochastic games, stationary Markov Nash equilibria (NE) are known to exist  \cite{Shapley1953,fink1964equilibrium}, and they are also stable against arbitrary history-dependent deviations. 
However, these equilibria are computationally intractable due to the \PPAD-hardness result of NE in normal-form games (\textit{i.e.,} single-state Markov games) \cite{daskalakis2009complexity, chen2009settling}.
Stationary Markov coarse correlated equilibria (CCE) also exist as a superset of stationary Markov NE. 
These equilibria provide a \textit{natural solution concept}: 
they give a compact, time-homogeneous description of behavior, and their incentive constraints can be expressed entirely in terms of the same stationary policy.

In normal-form games, CCE are computationally tractable: exact CCE can be computed efficiently via linear programming, and approximate CCE arise naturally from no-regret learning \cite{Aumann1974,CesaBianchiLugosi2006}. 
In Markov games, however, tractability depends sharply on the policy class.
Nonstationary non-Markovian CCE 
can be computed efficiently by algorithms such as \textsc{V-learning} \cite{JLWY21,SMB21,MB21}, but the resulting policies are generally not Markovian. 
Nonstationary Markov CCE are also tractable via backward induction, or via decentralized learning algorithms \cite{daskalakis2023complexity, cui2023breaking}.
By contrast, as established in \cite{daskalakis2023complexity,sidford_complexity,HansenNie2025}, computing stationary Markov CCE in general-sum stochastic games is \PPAD-complete, even with two players, turn-based control, constant discount factor, and constant approximation parameter.

A central subclass of stochastic games not captured by the above negative results is that of \emph{single-controller Markov games}, a model which goes back to classical work on stochastic games \cite{ParthasarathyRaghavan1981,FilarVrieze2012}. 
In these games, rewards may depend on the joint action profile, but the transition kernel depends only on the action of one distinguished player. 

Over the years, single-controller Markov games have been studied in communications, control, machine learning, and economics \cite{Basar1986,Eldosouky2016,FilarVrieze2012, Guan2016}. 
More recently, single-controller Markov games have emerged as a tractable frontier for learning and equilibrium computation under specific assumptions. 
In particular, fictitious play converges to stationary Markov NE in two-player zero-sum and identical-interest single-controller Markov games \cite{ozgaglar_controller,sc_icml};
single-controller structure allows efficient computation of nonstationary Markov NE in polymatrix zero-sum Markov games through an \textit{equilibrium-collapse phenomenon whereby the product marginals of CCE induce NE} \cite{kalogiannis_polymatrix};
and optimistic policy gradient converges to stationary Markov NE whenever such an equilibrium-collapse property holds \cite{panageasSingleController}. 
We defer a detailed literature review to Appendix \ref{appendix:related_work}.

However, beyond these assumptions, the computational landscape of finding stationary Markov CCE in single-controller stochastic games remains elusive, leading to the following
natural question:
\begin{quote}
    \textbf{Question 1: }\textit{What is the computational complexity of stationary Markov CCE in stochastic games with a single controller?}
\end{quote}
In addition, the reductions used in the hardness results of general-sum two-player Markov games \cite{daskalakis2023complexity,sidford_complexity,HansenNie2025} rely on \textit{turn-based} constructions in which each state is controlled by a single player, with control alternating across states, so that every player controls both rewards and transitions at some states. 
This structure effectively collapses Markov CCE to NE, thereby allowing hardness for NE to immediately transfer to CCE. 
That said, equilibrium collapse is not generally guaranteed in single-controller Markov games, and as discussed above, under equilibrium collapse finding stationary Markovian NE is tractable \cite{panageasSingleController}. 
This leaves open the following fundamental question:
\begin{quote}
    \textbf{Question 2: } \textit{Is the hardness of stationary Markov CCE an artifact of alternating transition control, or does it persist even when the state dynamics have single-controller structure?}
\end{quote}

\subsection{Contribution and significance}

In this paper, we resolve the aforementioned questions by showing that computing a stationary Markov $\Approx$-CCE in two-player stochastic games with a single controller is \PPAD-complete.
In particular, hardness is shown for both the \textit{perfect} notion (where the equilibrium conditions hold at every state) and the \textit{non-perfect} notion (where the equilibrium conditions hold in expectation under the law of the initial state distribution).
We prove the following main theorems:

\begin{theorem}[Abridged version of Theorem \ref{thm:perfect}]\label{thm: perfect_intro}
Computing a perfect stationary Markov $\Approx$-CCE in a discounted two-player single-controller Markov game is \PPAD-hard, even when the discount factor $\gamma$ is any fixed constant in  $[\nicefrac{1}{16}, 1),$\footnote{Please see also Remark \ref{rem:discount}.\label{fn:remark}} all rewards lie in $[0,1]$, and accuracy $\Approx>0$ is less than an absolute constant.
\end{theorem}

\begin{theorem}[Abridged version of Theorem \ref{thm:nonperfect}]\label{thm: non_perfect_intro}
Assume that \PCP-for-\PPAD for \textsc{Pure-Circuit} holds. Then, computing a non-perfect stationary Markov $\Approx$-CCE in a discounted two-player single-controller Markov game is \PPAD-hard, even when the discount factor $\gamma$ is any fixed constant in $[\nicefrac{1}{16}, 1)$,\footref{fn:remark} all rewards lie in $[0,1]$, and accuracy $\Approx>0$ is less than an absolute constant.
\end{theorem}

Similarly to \cite{daskalakis2023complexity}, we also provide an unconditional hardness result for the non-perfect notion (stated formally in Theorem \ref{thm:nonperfect-unconditional}), where the hardness holds for inverse polynomial accuracy parameter instead of constant.
Moreover, membership in \PPAD directly follows from the \PPAD-membership of approximate stationary Markov perfect NE in finite discounted general-sum stochastic games due to \cite{deng2023complexity}.
Since every stationary Markov perfect NE is a stationary Markov perfect CCE, the corresponding CCE search problem is in \PPAD.


Remarkably, our results imply that the \PPAD-hardness of stationary Markov CCE is robust to the single-controller restriction. 
The computational obstruction is more basic than previously understood:
the \PPAD-hardness results of stationary Markov CCE shown in prior work \cite{daskalakis2023complexity, HansenNie2025, sidford_complexity} relied on turn-based games which enforce equilibrium collapse phenomena, whereby hardness for NE immediately translated to CCE.
Our result shows that neither the switching-controller structure nor the equilibrium collapse property are what actually drives the hardness of stationary Markov CCE; it is enough that one player solely controls the transitions while both players affect rewards.\footnote{The joint control of rewards is a necessary condition for hardness in single-controller Markov games, since the restricted setting where a single player controls the transitions, and only the other player(s) control(s) the rewards, is in \P. This result is quite straightforward: the reward-controlling players compute a CCE strategy per-state, while the transition-controlling player solves a single-agent MDP problem given the CCE profile of the other players.}
In fact, it is impossible for a single-controller hard instance to have equilibrium-collapse because this would contradict the positive result of \cite{panageasSingleController}.
This renders the collapse property a tight computational barrier characterizing the complexity of stationary Markov CCE in single-controller stochastic games.

Based on the above, to the best of our knowledge, our result is the first to show \PPAD-hardness for computing CCE without relying on equilibrium collapse phenomena or reductions from hard Nash instances \cite{golowich-regret,anagnostides2023complexity, peng2024complexity}.
Instead, we construct single-controller gadgets whose local incentive constraints force the solution of a \textsc{Pure-Circuit} \cite{purecircuit} instance even under strongly correlated stationary policies (see also our technical overview in Section \ref{sec: technical_overview}).
To elaborate on the above, Section~\ref{sec:example} gives a formal example of a \textsc{Pure-Circuit} instance and a stationary Markov CCE of the corresponding reduced single-controller stochastic game. 
The decoded policy yields a valid solution of the circuit instance, but the CCE does not collapse to an approximate NE of the reduced game.

Furthermore, while no-regret learning has been shown to be intractable in several general Markov-game settings~\cite{golowich-regret,abbasi2013online,bai2020near}, the single-controller structure makes no-regret learning efficiently implementable.\footnote{This is due to the simple observation that only the transition-controlling player's stationary Markovian policy $\Policy_{\Controller}$ affects the induced discounted occupancy measure
$d(s,a_{\Controller})=(1-\Discount)\sum_{t\ge0}\Discount^t
\Pr^{\Policy_{\Controller}}[s_t=s,a_{\Controller,t}=a_{\Controller}\mid s_0\sim\InitialDist]$.
This is the state-action analogue of the visitation distribution defined in the preliminaries (Section \ref{sec: prelims}). The controller's utility is linear in $d$, and once $d$ is fixed, each non-controller's utility is linear in its own stationary policy. Hence, one can implement no-regret learning by letting the controller optimize over the occupancy-measure polytope~\cite{kakadeBook}, while the remaining players run no-regret algorithms over their stationary Markovian policy spaces.}
The empirical distribution of play is then a CCE of the normal-form game whose actions are stationary Markovian policies; equivalently, it is a sparse distribution over pure stationary Markovian policy profiles sampled by an initial common seed \cite{golowich-regret}.  
Such a CCE is stationary, since every sampled policy profile is stationary, but it is not itself a Markov CCE:
the initial seed can correlate behavior across different states and across different visits to the same state.  
Our result therefore identifies an interesting computational frontier in single-controller stochastic games.  
If one allows this slightly richer, seed-dependent policy class, no-regret learning can efficiently produce a CCE.  
If instead one insists that the correlated recommendation rule itself be stationary and Markovian, then computing a CCE is \PPAD-hard.

\input{attempts}

\input{technical_overview}

%% file: attempts.tex
\subsection{Where attempts for a positive result break down}

One may have expected that the single-controller structure could give a positive result for stationary Markov CCE, escaping the \PPAD-hardness appearing in switching-controller stochastic games. A natural first approach is to extend the classical framework of Hart and Schmeidler~\cite{hart1989existence} for normal-form games. This framework shows that the existence of (Coarse) Correlated Equilibria can be derived from a duality/minimax principle rather than from an explicit appeal to Brouwer's fixed point theorem.

The corresponding idea for single-controller stochastic games would be to exploit the fact that only one player controls the transition kernel. Since the non-controller's strategic role appears, at least superficially, through linear incentive constraints, one may try to view the controller together with the correlation device as the primal player of the minimax framework, and to introduce the controller's deviator as the dual adversary. This leads to a formulation reminiscent of a two-player zero-sum stochastic game: the correlating side chooses stationary recommendations and transition-inducing actions, while the deviator searches for a profitable deviation from the prescribed Markovian recommendation rule.
At first sight, this approach seems promising. Although the resulting problem is not convex-concave in any obvious static sense, one might hope to recover strong duality through the dynamic programming structure of discounted zero-sum stochastic games. Shapley's theorem for such games is based on the contraction of the discounted dynamic programming operator~\cite{Shapley1953}, and therefore ultimately on a Banach-type fixed point argument. If this program had succeeded with a well-conditioned contraction, this would have placed the problem in \CLS \cite{daskalakis2011continuous,daskalakis2018converse}.


The obstacle is that such a construction has constraints that are intrinsically coupled across states, which is not allowed in standard two-player zero-sum stochastic games. Consequently, the minimax formulation does not decompose into independent Hart--Schmeidler-type conditions state by state. The coupled constraints create extraneous solutions: spurious stationary points of the relaxed formulation need not correspond to stationary Markov CCE of the original stochastic game.
Apart from contraction arguments, gradient-domination/Polyak--{\L}ojasiewicz-type conditions are also not satisfied, which would otherwise allow one to get efficient convergence to the solution via policy-gradient-like methods.

Our hardness result explains why these attempts were too optimistic. 
The difficulty is not merely a technical inconvenience in the dual formulation; the stationary Markov CCE problem retains an unavoidable Brouwer-type fixed-point structure.
Thus, unlike the normal-form (Coarse) Correlated Equilibrium setting of Hart and Schmeidler, the single-controller stochastic-game setting does not seem to admit a straightforward existence or computation proof based only on strong duality, minimax contraction, or Banach fixed point arguments. 
The \PPAD-hardness result shows that the Brouwer-type nature of the problem is intrinsic: any approach proving existence of stationary Markov CCE in single-controller settings via a contraction/duality framework is unlikely to exist as it would have unexpected complexity-theoretic consequences (since $\CLS = \PPAD \cap \PLS$ \cite{fearnley2022complexity}, it would imply $\PPAD \subseteq \PLS$).

%% file: technical_overview.tex
\subsection{Technical overview} \label{sec: technical_overview}
The known hardness reductions \cite{daskalakis2023complexity,HansenNie2025} for stationary Markov CCE in two-player stochastic games (with binary actions) start from a circuit problem, such as \GeneralizedCircuit{} \cite{chen2009settling} or \PureCircuit{} \cite{purecircuit}. Such a circuit is a collection of variables connected by gates, where each gate imposes a local constraint on the values of its input and output variables. 
The goal is to find an assignment of values to all variables that approximately satisfies every gate. 
The common reduction recipe is to represent each circuit variable by a behavioral signal in the stationary policy, implement each gate by a local game gadget whose incentives enforce the corresponding constraint, and then glue the gadgets so that the output signal of one gate is used as the input signal of another. 

However, previous reductions use turn-based Markov games, which are well-suited for the application of the above recipe:
control alternation allows every state to have its own active controlling player, so the probability with which that player plays a specific action (e.g., action 1) provides a canonical signal, and the bipartite structure of the circuit can be used to glue gadgets while avoiding conflicts between the rewards assigned to the same player. 
Moreover, at a turn-based game the reward and transition depend only on the active player's action, thus satisfying equilibrium collapse where every CCE marginalizes to a NE. 
This allows such a reduction to construct a correlated policy which is a product distribution that affects the game only through the active player's policy. 
By contrast, the single-controller restriction removes precisely these simplifications. 
Since the same player controls the transitions for all states, gate inputs and outputs cannot be separated by changing state ownership. 
Furthermore, the equilibrium collapse property does not hold, so the equilibrium must be a genuinely correlated joint Markov policy.
This means that there is no a priori unique signal associated with a variable state, as the controller marginal, the other-player marginal, and the joint matching probability may all encode different information.
Thus, it is highly non-trivial how to encode the circuit assignments through such correlated policies.
Finally, CCE deviations in the single-controller setting are intrinsically less powerful than the correlated policies they are tested against: a deviating player must switch to an independent stationary policy against the other player's marginal \textit{for all states}, so the deviation cannot preserve any value that relies on correlation between the players' recommendations. 
Consequently, the task of implementing the gates becomes much more challenging.


Our reduction follows the same circuit-based recipe but introduces new machinery to tackle the challenges arising from the single-controller restriction. 
We reduce the \PPAD-complete problem \textsc{Pure-Circuit} to the problem of computing a stationary Markov $\Approx$-CCE in a two-player single-controller Markov game with binary actions. 
At each variable state, the controller's marginal action is the quantity that determines the next-state transition. 
We define the \textit{non-controller's marginal action} to be the quantity from which the Boolean value of the corresponding circuit node is decoded. 
The core task is therefore to synchronize these two marginals using only coarse-deviation constraints, without forcing the joint policy to become product. 
This is achieved through two {single-controller gadget conditions}: the \textit{copy gadget}~\hyperref[copy]{\gadgetC} and the \textit{controller-forcing gadget}~\hyperref[controller-forcing]{\gadgetF}.
The {copy gadget} forces an almost-Boolean controller marginal to be copied into the non-controller's marginal. 
The {controller-forcing gadget} uses gaps between the controller's two continuation values\footnote{The continuation value of the controller for selecting action $b$ is the expected value of the controller at the next state conditioned on the controller selecting action $b$ at the current state. It is formally defined in \eqref{eq:W-def}.} to force the controller marginal itself to be almost Boolean.
Once these two conditions hold, the decoded circuit assignment implements the NOT, OR, and PURIFY gates of the \textsc{Pure-Circuit} instance (see Proposition \ref{prop:gate-soundness}).

Crucially, unlike turn-based reductions, transition control is assigned to a single player throughout the game, rather than changing from state to state. Thus the construction cannot give input and output gate nodes different strategic roles through state ownership or gate-specific reward tables. Instead, all variable states use the same reward functions for both players, while the gate logic is implemented through the controller's transition targets and the sink values.
In Figure \ref{fig:gates}, we provide a graphical illustration of our construction along with the defined reward functions of both players (formally defined in \eqref{eq:lambda-reward}).
More specifically, the non-controller player receives reward $1$ exactly when its action matches the controller's action; this is what makes a failed copy condition expose a profitable coarse deviation for the non-controller. 
The controller receives a reward determined primarily by the non-controller's action, so that its own immediate action does not directly affect its immediate reward. 
Consequently, the controller's incentives are driven by continuation-value differences rather than by immediate payoffs. 
Finally, the absorbing sink states give the controller fixed continuation values, which are chosen so that, in each NOT, OR, and PURIFY gate, the controller is incentivized to select the output action prescribed by the corresponding gate constraint.

To make our construction work for both perfect and non-perfect notions, we parameterize the controller's variable-state reward by a constant $\lambda \ge 0$. 
For the perfect Markov CCE hardness result, we set $\lambda=0$, so the controller's immediate reward at variable states is independent of its own action. Moreover, because the perfect CCE constraints hold from every initial state, each gadget condition can be proved state by state: a violation at a single variable state directly yields a profitable deviation from that state.
The non-controller player cannot affect transitions, so if its marginal fails to copy an almost-deterministic controller marginal, a one-state coarse deviation would yield an immediate matching-payoff gain that dominates any possible discounted continuation loss. 
This gives the copy gadget condition for the perfect case (see Lemma \ref{lem:perfect-copy}).
For the controller, its immediate reward at variable states is independent of its own action. Hence its incentive between the two actions at an output state comes only from their continuation values. 
A continuation-value gap therefore creates a local action advantage (see Lemma \ref{lem:best-adv-perfect}); the perfect CCE inequalities, together with a performance-difference argument \cite{agarwal2021theory}, imply that such advantages cannot remain large on the action that the controller plays with substantial probability. 
This yields the controller-forcing condition (see Lemma \ref{lem:perfect-force}) and, consequently, an unconditional constant-accuracy hardness result for perfect Markov CCE (see Theorem \ref{thm:perfect}).

The non-perfect case is more subtle because the equilibrium inequalities hold only in expectation from the initial state distribution. 
In the forcing argument, the controller deviation can be chosen so that its local advantage is nonnegative at every state: it plays the continuation-better action on forcing-error states and preserves the original controller marginal elsewhere. Hence the positive gains from forcing errors cannot be canceled by losses at other states. The copy argument is different. The reward player's best-matching deviation may gain on copy-error states but lose on well-correlated non-error states, so a gain--loss cancellation can occur when the non-perfect CCE inequality is averaged from the initial distribution. The $\lambda$-bonus supplies a controller decorrelation deviation whose gain cancels the dependence on the actual joint matching probability, replacing it by a nonnegative marginal witness of copy failure.
The bonus is chosen small enough to preserve the intended value ordering inside each gate, while ensuring that copy failures cannot be averaged away through correlation.
Intuitively, if the copy condition fails on many visited states, then either the non-controller can profit by switching to a better matching action, or correlation is masking this failure; in the latter case, the controller can profit from the decorrelation bonus by resampling its own action according to its marginal.
Lemma~\ref{lem:decorrelation-identity} formalizes the above incentives, and Lemma~\ref{lem:average-copy} turns it into an averaged copy condition.
A separate averaged controller-forcing argument (see Lemma \ref{lem:average-forcing}) then shows that a large mass of continuation-gap violations gives the controller a profitable stationary deviation. 
Finally, using the above average gadget arguments, Proposition~\ref{prop:bad-gates} converts small mass of copy and forcing failures into a bound on the number of unsatisfied \textsc{Pure-Circuit} gates.
This gives the conditional constant-accuracy non-perfect hardness theorem (Theorem~\ref{thm:nonperfect}) under \PCP-for-\PPAD for \PureCircuit{} \cite{deligkas2026fisher} (see Conjecture~\ref{conj:pcp-ppad}), and the same averaging framework yields the unconditional inverse-polynomial theorem (Theorem~\ref{thm:nonperfect-unconditional}).

%% file: example.tex
\section{Example: Pure-Circuit solution with no equilibrium collapse}
\label{sec:example}

\begin{figure}[h]
\centering
\begin{tikzpicture}[
    scale=.94,
    transform shape,
    x=1cm,
    y=1cm,
    every node/.style={font=\small},
    examplesource/.style={
        draw=ownerA,
        very thick,
        rounded corners=2pt,
        minimum width=18mm,
        minimum height=10mm,
        align=center,
        fill=white,
        inner sep=2pt
    },
    exampleoutput/.style={
        draw=ownerA,
        very thick,
        rounded corners=2pt,
        minimum width=18mm,
        minimum height=10mm,
        align=center,
        fill=white,
        inner sep=2pt
    },
    examplegate/.style={
        draw=black,
        rounded corners=2pt,
        minimum width=18mm,
        minimum height=8mm,
        align=center,
        fill=white,
        inner sep=2pt
    },
    examplewire/.style={-{Latex[length=2mm]}, very thick}
]

\node[examplesource] (a) at (0,3.0)
  {\begin{tabular}{c}$a$\\$\Decoded{a}=0$\end{tabular}};
\node[examplegate] (gnot) at (2.3,3.0) {\NOT};
\node[exampleoutput] (d) at (4.6,3.0)
  {\begin{tabular}{c}$d$\\$\Decoded{d}=1$\end{tabular}};

\node[examplesource] (b) at (4.6,0.0)
  {\begin{tabular}{c}$b$\\$\Decoded{b}=\bot$\end{tabular}};
\node[examplegate] (gor) at (7.2,1.5) {\OR};
\node[exampleoutput] (e) at (9.8,1.5)
  {\begin{tabular}{c}$e$\\$\Decoded{e}=1$\end{tabular}};

\node[examplegate,minimum width=23mm] (gpur) at (12.4,1.5) {\PURIFY};
\node[exampleoutput] (hzero) at (15.1,2.7)
  {\begin{tabular}{c}$h_0$\\$\Decoded{h_0}=1$\end{tabular}};
\node[exampleoutput] (hone) at (15.1,0.3)
  {\begin{tabular}{c}$h_1$\\$\Decoded{h_1}=1$\end{tabular}};

\draw[examplewire] (a) -- (gnot);
\draw[examplewire] (gnot) -- (d);
\draw[examplewire] (d.east) to[out=0,in=155] (gor.north west);
\draw[examplewire] (b.east) to[out=0,in=205] (gor.south west);
\draw[examplewire] (gor) -- (e);
\draw[examplewire] (e) -- (gpur);
\draw[examplewire] (gpur.north east) to[out=20,in=180] (hzero.west);
\draw[examplewire] (gpur.south east) to[out=-20,in=180] (hone.west);

\end{tikzpicture}
\caption{The three-gate \PureCircuit{} instance in
\eqref{eq:small-example-circuit}.  Red rounded rectangles are circuit
variables.  The labels show the decoded assignment induced by the policy below.}
\label{fig:small-mixed-cce-circuit}
\end{figure}
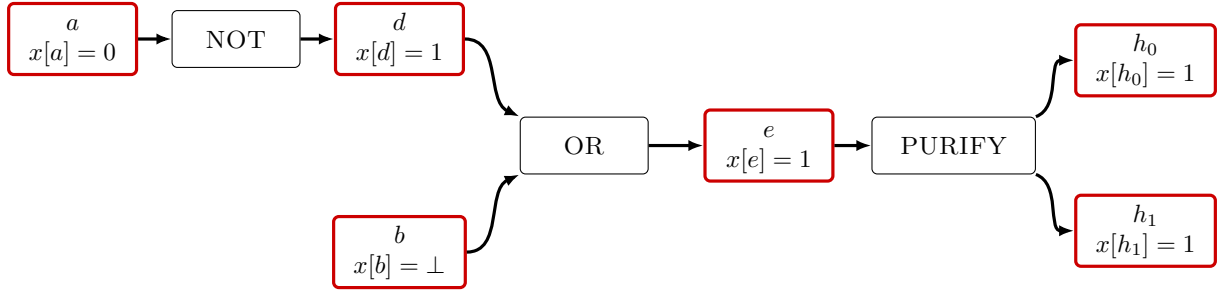

We give a small instance of the construction that illustrates two features of stationary Markov CCE.  First, the joint policy may remain genuinely correlated at individual states; the equilibrium constraints do not force it to collapse to the product of its one-player marginals.  Second, the same stationary policy is an exact perfect Markov CCE for $G_0(I)$ and, for positive $\lambda$, an exact non-perfect Markov CCE from the output-supported initial distribution.

\paragraph{The circuit and the Markov-game instance.}
Consider the connected three-gate circuit
\begin{equation}
d=\NOT(a), \qquad e=\OR(d,b), \qquad (h_0,h_1)=\PURIFY(e).
\label{eq:small-example-circuit}
\end{equation}
Its node set and output-state set are
\begin{equation}
\NodeSet=\{a,b,d,e,h_0,h_1\}, \qquad
\OutputStates=\{\State{d},\State{e},\State{h_0},\State{h_1}\}.
\label{eq:small-example-output-set}
\end{equation}
Figure~\ref{fig:small-mixed-cce-circuit} displays the circuit in a left-to-right layout. 
The gate boxes separate converging and diverging wires, so no two arrows overlap.

We fix $\Discount=\FixedDiscount$ and construct the two-player game $G_\lambda(I)$ from Section~\ref{sec:construction}. 
Thus the rewards at every variable state are those in \eqref{eq:lambda-reward}. 
The source states $\State{a}$ and $\State{b}$ are not outputs of gates, so both controller actions at either state lead to $\Sink{0}$.  The remaining transitions are shown in Figure~\ref{fig:small-mixed-cce-markov-game}.

\begin{figure}[t]
\centering
\begin{tikzpicture}[
    x=1cm,
    y=1cm,
    every node/.style={font=\small},
    exampleoutputstate/.style={
        draw=ownerA,
        very thick,
        rounded corners=2pt,
        minimum width=18mm,
        minimum height=8mm,
        align=center,
        fill=white
    },
    examplesourcestate/.style={
        draw=black,
        fill=ownerBfill,
        very thick,
        circle,
        minimum size=10mm,
        align=center
    },
    examplesink/.style={
        draw=black,
        fill=sinkfill,
        rounded corners=2pt,
        minimum width=20mm,
        minimum height=8mm,
        align=center
    },
    examplezeroedge/.style={-{Latex[length=2mm]}, very thick, actionzero},
    exampleoneedge/.style={-{Latex[length=2mm]}, very thick, actionone},
    examplebothedge/.style={-{Latex[length=2mm]}, thick, dotted},
    exampleprob/.style={font=\scriptsize, fill=white, inner sep=1pt, text=black},
    gadgetlabel/.style={
        font=\scriptsize\bfseries,
        rounded corners=2pt,
        fill=white,
        inner sep=2pt
    }
]


\filldraw[
    rounded corners=7pt,
    draw=Purple!80!black,
    line width=1.1pt,
    fill=Purple!22,
    fill opacity=.24
] (-7.20,6.70) rectangle (7.20,3.45);

\filldraw[
    rounded corners=7pt,
    draw=DarkGreen!80!black,
    line width=1.1pt,
    fill=DarkGreen!22,
    fill opacity=.22
] (-6.55,4.35) rectangle (4.10,1.25);

\filldraw[
    rounded corners=7pt,
    draw=LightBlue!80!black,
    line width=1.1pt,
    fill=LightBlue!30,
    fill opacity=.24
] (-6.35,2.35) rectangle (-1.15,-0.75);


\node[exampleoutputstate] (hzero) at (-3.4,5.95) {$\State{h_0}$};
\node[exampleoutputstate] (hone)  at ( 3.4,5.95) {$\State{h_1}$};
\node[examplesink] (zhzero) at (-6.0,5.25) {$\Sink{89/128}$};
\node[examplesink] (zhone)  at ( 6.0,5.25) {$\Sink{39/128}$};

\node[exampleoutputstate] (e) at (0,3.75) {$\State{e}$};
\node[examplesink] (zor) at (-5.0,2.9) {$\Sink{21/64}$};

\node[exampleoutputstate] (d) at (-2.6,1.75) {$\State{d}$};
\node[examplesourcestate] (b) at (2.6,1.95) {$\State{b}$};
\node[examplesink] (zhalf) at (-5.0,0.65) {$\Sink{1/2}$};

\node[examplesourcestate] (a) at (-2.6,-0.00) {$\State{a}$};
\node[examplesink] (zzero) at (0,-1.85) {$\Sink{0}$};


\draw[examplezeroedge] (hzero.west) to[out=200,in=35] (zhzero.north east);
\draw[exampleoneedge] (hzero.south east) to[out=-35,in=155] (e.north west);
\draw[examplezeroedge] (hone.east) to[out=-20,in=145] (zhone.north west);
\draw[exampleoneedge] (hone.south west) to[out=-145,in=25] (e.north east);

\draw[examplezeroedge] (e.west) to[out=190,in=25] (zor.north east);
\draw[exampleoneedge] (e.south west) to[out=-135,in=70]
  node[exampleprob,pos=.55,left] {$1/2$} (d.north);
\draw[exampleoneedge] (e.south east) to[out=-45,in=110]
  node[exampleprob,pos=.55,right] {$1/2$} (b.north);

\draw[examplezeroedge] (d.south) -- (a.north);
\draw[exampleoneedge] (d.west) to[out=180,in=20] (zhalf.east);

\draw[examplebothedge] (a.south east) to[out=-45,in=160]
  node[exampleprob,pos=.55,below left] {$0,1$} (zzero.north west);
\draw[examplebothedge] (b.south west) to[out=-135,in=20]
  node[exampleprob,pos=.55,below right] {$0,1$} (zzero.north east);


\node[gadgetlabel, text=Purple!80!black] at (0,6.38) {\PURIFY{} GATE};
\node[gadgetlabel, text=DarkGreen!80!black] at (0,2.55) {\OR{} GATE};
\node[gadgetlabel, text=LightBlue!80!black] at (-5.00,-0.35) {\NOT{} GATE};

\end{tikzpicture}
\caption{The transition graph of the example.  Orange arrows correspond to controller action $0$, blue arrows to controller action $1$, and dotted arrows indicate that both controller actions lead to the same sink.  The translucent boxes highlight the \NOT{}, \OR{}, and \PURIFY{} gates.  The red states are precisely the output states in \eqref{eq:small-example-output-set}.}
\label{fig:small-mixed-cce-markov-game}
\end{figure}
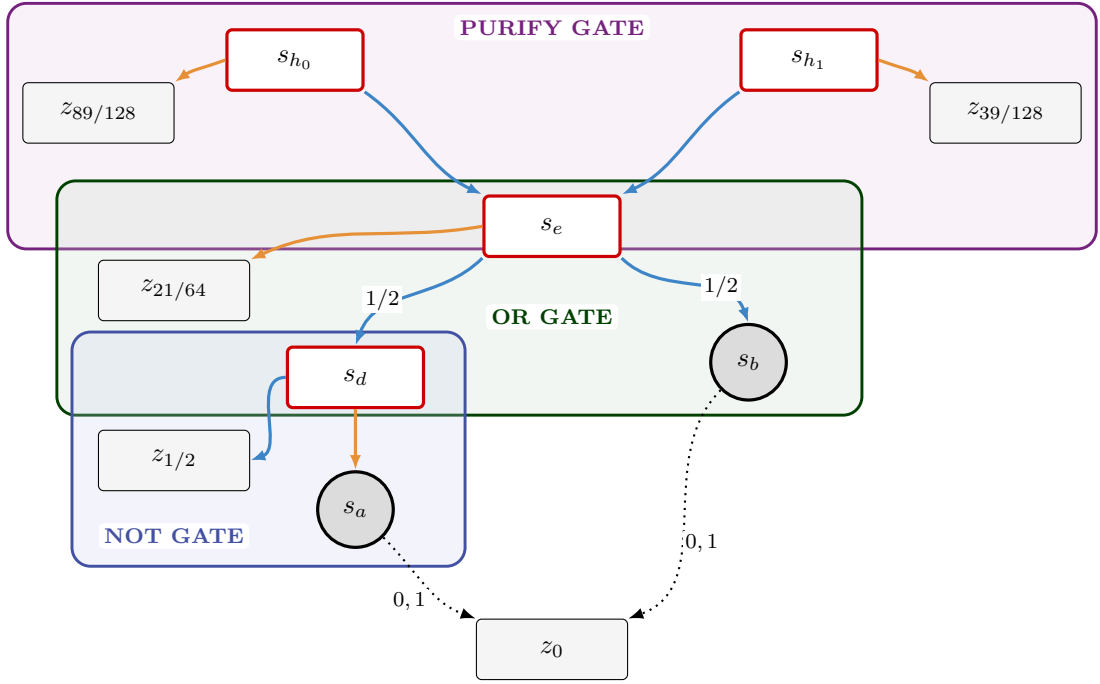

\paragraph{The stationary correlated policy.}
For each variable state $\State{u}$, we write
\begin{equation}
\Policy_u(i,j)
:=
\Policy(\Action{\Controller}=i,\Action{\RewardPlayer}=j\mid\State{u}). \nonumber
\label{eq:small-example-policy-entry}
\end{equation}
At the two source states, we define
\begin{equation}
\begin{array}{c|cccc|ccc}
 u
 & \Policy_u(0,0)
 & \Policy_u(0,1)
 & \Policy_u(1,0)
 & \Policy_u(1,1)
 & \RewMarg{u}
 & \CtrlMarg{u}
 & \MatchProb{u}
\\
\hline
 a & 4/5 & 0   & 0   & 1/5 & 1/5 & 1/5 & 1   \\
 b & 1/4 & 1/4 & 1/4 & 1/4 & 1/2 & 1/2 & 1/2
\end{array}
\label{eq:small-example-source-policy}
\end{equation}
and at every output state let
\begin{equation}
\Policy((1,1)\mid\State{d})
=
\Policy((1,1)\mid\State{e})
=
\Policy((1,1)\mid\State{h_0})
=
\Policy((1,1)\mid\State{h_1})=1.
\label{eq:small-example-output-policy}
\end{equation}
At every sink, $\Policy$ plays the unique dummy action.  The policy is genuinely correlated at $\State{a}$, where it mixes between $(0,0)$ and $(1,1)$.  At $\State{b}$, however, it is the product of two uniform marginals. This product row is what prevents an on-path controller decorrelation gain when
$\lambda>0$.

By \eqref{eq:decode}, the signal marginals in
\eqref{eq:small-example-source-policy}--\eqref{eq:small-example-output-policy} decode as
\begin{equation}
\begin{array}{c|cccccc}
 u & a & b & d & e & h_0 & h_1 \\
\hline
\RewMarg{u} & 1/5 & 1/2 & 1 & 1 & 1 & 1 \\
\Decoded{u} & 0 & \bot & 1 & 1 & 1 & 1
\end{array}.
\label{eq:small-example-decoding}
\end{equation}
The assignment in \eqref{eq:small-example-decoding} satisfies every gate in \eqref{eq:small-example-circuit}: the \NOT{} gate maps $0$ to $1$; the \OR{} gate has the Boolean input $\Decoded{d}=1$ and therefore valid output $\Decoded{e}=1$; and the two \PURIFY{} outputs both equal the Boolean input $\Decoded{e}=1$.

\begin{proposition}
\label{prop:small-mixed-cce-example}
The stationary joint policy $\Policy$ defined in
\eqref{eq:small-example-source-policy} and
\eqref{eq:small-example-output-policy} has the following properties:
\begin{enumerate}[label=\textup{(\roman*)},leftmargin=8mm]
\item For $\lambda=0$, it is an exact perfect Markov CCE of $G_0(I)$.
\item For every $0<\lambda<1/180$, it is an exact non-perfect Markov CCE of
$G_\lambda(I)$ from the distribution $\InitialDist$ that is uniform on
$\OutputStates$.
\item For every $\lambda>0$, it is not a perfect Markov $\Approx$-CCE, for $\Approx < \frac{3\lambda}{4}$.
\end{enumerate}
\end{proposition}

\begin{proof}
We check deviations directly from the definition of value.

\paragraph{Reward player.}
Let $\DeviationPolicy{\RewardPlayer}$ be any stationary deviation and set
$\alpha_s:=\DeviationPolicy{\RewardPlayer}(1\mid s)$.  At a variable state
$s$, its expected matching payoff is
\begin{equation}
\alpha_s\CtrlMarg{s}+(1-\alpha_s)(1-\CtrlMarg{s})
\le \max\{\CtrlMarg{s},1-\CtrlMarg{s}\}. \nonumber
\label{eq:small-example-direct-reward-bound}
\end{equation}
For $\State{a}$, $\State{b}$, and the four output states, respectively, the
pairs
$\bigl(\MatchProb{s},\max\{\CtrlMarg{s},1-\CtrlMarg{s}\}\bigr)$
are
\[
(1,4/5),
\qquad
(1/2,1/2),
\qquad
(1,1).
\]
Thus the deviation cannot increase the reward player's expected payoff at any state.  Its action does not affect transitions, so the deviated and original state trajectories have the same law.  Therefore, for every initial distribution $\nu$, it holds that
\begin{equation}
\begin{aligned}
\Value{\RewardPlayer}{\DeviationProfile{\RewardPlayer}}{\nu}
-\Value{\RewardPlayer}{\Policy}{\nu}
&=
(1-\Discount)\Expect_{\nu}^{\Policy}
\left[
\sum_{t\ge0}\Discount^t
\left(
\alpha_{s_t}\CtrlMarg{s_t}
+(1-\alpha_{s_t})(1-\CtrlMarg{s_t})
-\MatchProb{s_t}
\right)
\right]\le 0,
\end{aligned}
\label{eq:small-example-direct-reward-global}
\end{equation}
where the summand is $0$ at sinks.

\paragraph{Controller.}
Against the fixed reward-player marginal, we define
\begin{equation}
\bar r_{\Controller}(s,b)
:=
\sum_{j\in\Bits}\OtherPolicy{\Controller}(j\mid s)
\left[(1-\lambda)j+\lambda\ind[b\ne j]\right].
\label{eq:small-example-direct-controller-stage} \nonumber
\end{equation}
The non-sink transition graph is acyclic.  Hence the largest value obtainable from a state is given, by backward induction, by
\begin{equation}
\begin{aligned}
B(\Sink{\theta})&=\theta,\\
B(s)&=
\max_{b\in\Bits}
\left\{
(1-\Discount)\bar r_{\Controller}(s,b)
+
\Discount\sum_{s'}\TransKernel(s'\mid s,b)B(s')
\right\}.
\end{aligned}
\label{eq:small-example-direct-bellman}
\end{equation}
Indeed, a mixed deviating action only forms a convex combination of the two quantities in \eqref{eq:small-example-direct-bellman}.  Thus every stationary controller deviation has value at most $B(s)$ from $s$.

For controller actions $(0,1)$, the corresponding values of the one-step payoff term $\bar r_{\Controller}(s,b)$ are
\[
(1/5,1/5+3\lambda/5)
\text{ at }\State{a},
\qquad
(1/2,1/2)
\text{ at }\State{b},
\qquad
(1,1-\lambda)
\text{ at every output state}.
\]
Consequently, we have
\begin{equation}
\Value{\Controller}{\Policy}{\State{a}}
=\frac3{16}(1-\lambda),
\qquad
B(\State{a})=\frac3{16}+\frac{9\lambda}{16},
\qquad
\Value{\Controller}{\Policy}{\State{b}}
=B(\State{b})=\frac{15}{32}.
\label{eq:small-example-direct-source-values}
\end{equation}
Evaluating \eqref{eq:small-example-direct-bellman} successively at
$\State{d},\State{e},\State{h_0},\State{h_1}$, define the gap at a state
$s$ to be the Bellman value of controller action $1$ minus the Bellman value of
controller action $0$ at $s$.  These gaps are
\begin{equation}
\begin{array}{c|c}
 s & \text{gap} \\
\hline
\State{d}   & (5-249\lambda)/256 \\
\State{e}   & (25-990\lambda)/1024 \\
\State{h_0} & (147-8175\lambda)/8192 \\
\State{h_1} & (347-8175\lambda)/8192
\end{array}. \nonumber
\label{eq:small-example-direct-action-gaps}
\end{equation}
All four gaps are positive for $0\le\lambda<1/180$.  Thus action $1$ is optimal at every output state.  Since $\Policy$ plays action $1$ there, and the action-$1$ successors are, successively, sinks, $\State{d}$ and $\State{b}$, and $\State{e}$, backward induction gives
\begin{equation}
B(s)=\Value{\Controller}{\Policy}{s}
\qquad
\text{for every }s\in\OutputStates.
\label{eq:small-example-direct-output-optimal}
\end{equation}
Let $\InitialDist$ be uniform on $\OutputStates$.  For every stationary controller deviation,
\begin{equation}
\Value{\Controller}{\DeviationProfile{\Controller}}{\InitialDist}
\le
\sum_{s\in\OutputStates}\InitialDist(s)B(s)
=
\Value{\Controller}{\Policy}{\InitialDist}. \nonumber
\label{eq:small-example-direct-nonperfect}
\end{equation}
Together with \eqref{eq:small-example-direct-reward-global}, this proves
part~\textup{(ii)}.

When $\lambda=0$, \eqref{eq:small-example-direct-source-values} also gives $B(\State{a})=\Value{\Controller}{\Policy}{\State{a}}$; equality already holds at $\State{b}$, at every output state by \eqref{eq:small-example-direct-output-optimal}, and at every sink.  Hence no controller deviation improves from any state.  Together with \eqref{eq:small-example-direct-reward-global}, this proves part~\textup{(i)}.

Finally, for every $\lambda>0$, deviating to action $1$ from $\State{a}$ gives
\begin{equation}
B(\State{a})-
\Value{\Controller}{\Policy}{\State{a}}
=\frac{3\lambda}{4}. \nonumber
\label{eq:small-example-direct-not-perfect}
\end{equation} 
Thus the policy is  not a perfect Markov $\Approx$-CCE, for $\Approx < \frac{3\lambda}{4}$, proving part~\textup{(iii)}.
\end{proof}

\paragraph{Failure of equilibrium collapse.}
The correlation at $\State{a}$ is essential.  Let $\widehat{\Policy}$ be the product policy with the same one-player marginals as $\Policy$.  Both marginals
at $\State{a}$ put probability $1/5$ on action $1$, so under
$\widehat{\Policy}$ the reward player's matching probability is
\begin{equation}
\ProdMatch{a}
=
\left(1-\frac15\right)^2+
\left(\frac15\right)^2
=
\frac{17}{25}. \nonumber
\label{eq:small-example-product-match}
\end{equation}
A fixed deviation to action $0$ instead matches with probability
$1-\CtrlMarg{a}=4/5$.  Because both controller actions at $\State{a}$ lead to $\Sink{0}$, this yields normalized value gain
\begin{equation}
(1-\Discount)
\left(\frac45-\frac{17}{25}\right)
=
\frac9{80}>0. \nonumber
\label{eq:small-example-product-deviation}
\end{equation}
Thus the product policy with the same marginals is not even a reward-player best response at $\State{a}$.  The example therefore exhibits a genuinely correlated stationary Markov CCE rather than an equilibrium obtained by independent mixing.

%% file: conclusion.tex
\section{Conclusion}

In this paper, we prove that computing approximate stationary perfect Markov CCE in discounted single-controller stochastic games is \PPAD-complete, even with two players, binary actions, a fixed discount factor, and a constant approximation parameter. 
Our results show that the hardness of stationary Markov CCE is not an artifact of switching-controller transition setup or equilibrium-collapse phenomena: it persists even when a single player controls all state transitions and joint policies are highly correlated. For the non-perfect notion, we obtain constant-accuracy hardness under the \PCP-for-\PPAD hypothesis and unconditional hardness for inverse-polynomial accuracy.

\medskip

\noindent
There are several important open questions for follow-up research:

\begin{itemize}
    \item While our result settles the computational landscape of stationary Markov CCE, it remains a compelling open question whether efficient no-regret learning is possible in discounted general-sum stochastic games.
    In particular, although no-regret learning has been shown to be intractable in many Markov game settings \cite{golowich-regret,abbasi2013online,bai2020near}, it remains unclear whether there exist efficient learning algorithms which play stationary Markovian policies and are no-regret against stationary Markovian benchmarks in general-sum discounted stochastic games.
    \item Furthermore, our paper leaves open the complementary statistical question of identifying the optimal sample complexity of learning nonstationary Markov CCE and CE.
    In particular, while nonstationary Markov CCE and CE are known to be learnable with $\mathcal{O}(1/\varepsilon^2)$ sample complexity~\cite{cui2023breaking}, the optimal sample complexity remains unclear.
    The same question is also interesting for the single-controller Markov game setting.
\end{itemize}

%% file: appendix.tex
\appendix

\section{Extended Related Work}\label{appendix:related_work}

\paragraph{Equilibrium computation in Markov games.}

The one-state, one-stage special case of a Markov game is precisely a normal-form game. 
In this setting, the computational landscape is well understood: Nash equilibria in two-player zero-sum games can be computed efficiently via linear programming \cite{v1928theorie}, whereas Nash equilibria in general-sum games are \PPAD-complete \cite{daskalakis2009complexity,chen2009settling}. 
By contrast, correlated equilibria (CE) and coarse correlated equilibria (CCE) admit linear feasibility descriptions and can be computed efficiently in explicitly represented games \cite{Aumann1974}; moreover, correlated equilibria arise naturally as the outcome of decentralized no-regret learning dynamics~\cite{CesaBianchiLugosi2006,stoltz2005internal,blum2007external}.

Moving from this static benchmark to genuine multi-state Markov games, tractability depends sharply on the policy class.
Nonstationary non-Markovian CCE can be computed efficiently by algorithms such as \textsc{V-learning} \cite{JLWY21,SMB21,MB21}, but the resulting policies are generally not Markovian. 
Nonstationary Markov CCE are also tractable via backward induction, or via decentralized learning algorithms \cite{cui2023breaking, daskalakis2009complexity}.
By contrast, as established in \cite{daskalakis2023complexity,sidford_complexity,HansenNie2025}, computing stationary Markov CCE in general-sum stochastic games is \PPAD-complete, even with two players, constant discount factor and constant approximation parameter. 
The reductions used in these works, however, rely on turn-based constructions in which each state is controlled by a single player, with control alternating across states, so that every player controls both rewards and transitions at some states.
This structure effectively collapses Markov CCE to  Markov NE, thereby allowing hardness for NE to immediately transfer to CCE. 
A key question left open is whether the hardness persists for single-controller stochastic games, where Markov CCE do not collapse to Nash equilibria, or whether it crucially relies on both players being able to affect the state dynamics.

\paragraph{Single-controller Markov games.}
Single-controller Markov games constitute a central subclass of stochastic games, in which rewards may depend on the joint action profile, but the transition kernel depends only on the action of one distinguished player.
This model goes back to classical work on stochastic games \cite{ParthasarathyRaghavan1981,FilarVrieze2012}, and captures dynamic strategic settings in which one player controls state evolution while all players may affect payoffs. 
Classical algorithmic work on single-controller stochastic games developed linear-programming formulations for zero-sum cases and finite-step characterizations of stationary Markov Nash equilibria in the general-sum case via matrix-game, quadratic-programming, and LCP/Lemke-type reductions \cite{sc1, sc2, sc3, sc4}; however, these results are finite-step or global-optimization characterizations rather than polynomial-time algorithms for  equilibrium computation, and they do not address correlated equilibria.
Over the years, single-controller Markov games have been studied in communications, control, machine learning, and economics \cite{Basar1986,Eldosouky2016,FilarVrieze2012, Guan2016}. 
More recently, single-controller Markov games have emerged as a tractable frontier for learning and equilibrium computation under specific assumptions. 
In particular, fictitious-play-type dynamics converge to stationary Markov NE in two-player zero-sum and many-player identical-interest single-controller Markov games \cite{ozgaglar_controller, sc_icml};
single-controller structure allows efficient computation of nonstationary Markov NE in polymatrix zero-sum Markov games through an equilibrium-collapse phenomenon whereby CCE induce NE \cite{kalogiannis_polymatrix};
and optimistic policy-gradient methods converge to stationary Markov NE in single-controller stochastic games whenever such an equilibrium-collapse property holds \cite{panageasSingleController}.

\section{Proof of Lemma \ref{lem:performance-difference-deviation-policy}}\label{appendix:performance}

\begin{lemma*}[Performance-difference  lemma for one-player deviations]
Fix a stationary joint policy $\Policy$ and a deviation
policy $\DeviationPolicy{i}$ for player $i$.  
Then
\begin{equation}
\Value{i}{\DeviationProfile{i}}{\State{0}}-\Value{i}{\Policy}{\State{0}}
=
\frac{1}{1-\Discount}
\sum_{s\in\StateSpace}
\VisitDist{\State{0}}{\DeviationProfile{i}}{s}
\CoarseAdvantage{i}{\Policy}{\DeviationPolicy{i}}{s}.
\label{eq:coarse-perf-diff_}
\end{equation}
\end{lemma*}

\begin{proof}
Define
\[
D(s):=\Value{i}{\DeviationProfile{i}}{s}-\Value{i}{\Policy}{s}.
\]
Also define the transition kernel induced by $\DeviationProfile{i}$ as
\[
\TransKernel^{\DeviationProfile{i}}(s'\mid s)
:=
\sum_{b\in\ActionProfileSet{s}}
(\DeviationProfile{i})(b\mid s)\TransKernel(s'\mid s,b).
\]
By the Bellman equation for the deviated policy $\DeviationProfile{i}$,
\[
\Value{i}{\DeviationProfile{i}}{s}
=
\sum_{b\in\ActionProfileSet{s}}
(\DeviationProfile{i})(b\mid s)
\left[
(1-\Discount)\RewardSymbol_i(s,b)
+
\Discount\sum_{s'\in\StateSpace}\TransKernel(s'\mid s,b)
\Value{i}{\DeviationProfile{i}}{s'}
\right].
\]
Add and subtract $\Value{i}{\Policy}{s'}$ inside the continuation term.  Then
\[
\begin{aligned}
\Value{i}{\DeviationProfile{i}}{s}
&=
\sum_{b\in\ActionProfileSet{s}}
(\DeviationProfile{i})(b\mid s)
\left[
(1-\Discount)\RewardSymbol_i(s,b)
+
\Discount\sum_{s'\in\StateSpace}\TransKernel(s'\mid s,b)
\Value{i}{\Policy}{s'}
\right]
\\
&\quad+
\Discount
\sum_{b\in\ActionProfileSet{s}}
(\DeviationProfile{i})(b\mid s)
\sum_{s'\in\StateSpace}
\TransKernel(s'\mid s,b)
\left(
\Value{i}{\DeviationProfile{i}}{s'}-\Value{i}{\Policy}{s'}
\right).
\end{aligned}
\]
By the definition of $\QValue{i}{\Policy}{s,b}$, this gives
\[
\Value{i}{\DeviationProfile{i}}{s}
=
\sum_{b\in\ActionProfileSet{s}}
(\DeviationProfile{i})(b\mid s)\QValue{i}{\Policy}{s,b}
+
\Discount
\sum_{s'\in\StateSpace}
\TransKernel^{\DeviationProfile{i}}(s'\mid s)D(s').
\]
Subtracting
\[
\Value{i}{\Policy}{s}=
\sum_{a\in\ActionProfileSet{s}}
\Policy(a\mid s)\QValue{i}{\Policy}{s,a}
\]
and using the definition of $\CoarseAdvantage{i}{\Policy}{\DeviationPolicy{i}}{s}$ yields the
one-step identity
\[
D(s)
=
\CoarseAdvantage{i}{\Policy}{\DeviationPolicy{i}}{s}
+
\Discount
\sum_{s'\in\StateSpace}
\TransKernel^{\DeviationProfile{i}}(s'\mid s)D(s').
\]
Equivalently,
\[
D(s)
=
\CoarseAdvantage{i}{\Policy}{\DeviationPolicy{i}}{s}
+
\Discount\Expect[D(\State{1})\mid \State{0}=s,\DeviationProfile{i}].
\]
Iterating this identity from $\State{0}$, for every integer $T\ge1$ we obtain
\[
D(\State{0})
=
\Expect\left[
\sum_{t=0}^{T-1}\Discount^t
\CoarseAdvantage{i}{\Policy}{\DeviationPolicy{i}}{\State{t}}
+
\Discount^T D(\State{T})
\;\middle|\;
\State{0},\DeviationProfile{i}
\right].
\]
Because rewards lie in $[0,1]$, all normalized discounted values lie in
$[0,1]$.  Hence $|D(\State{T})|\le1$, so
\[
\left|
\Discount^T\Expect[D(\State{T})\mid \State{0},\DeviationProfile{i}]
\right|
\le \Discount^T\to0.
\]
Letting $T\to\infty$ gives
\[
D(\State{0})
=
\Expect\left[
\sum_{t\ge0}\Discount^t
\CoarseAdvantage{i}{\Policy}{\DeviationPolicy{i}}{\State{t}}
\;\middle|\;
\State{0},\DeviationProfile{i}
\right].
\]
Since the state space is finite, this may be rewritten as
\[
\begin{aligned}
D(\State{0})
&=
\sum_{t=0}^{\infty}\Discount^t
\sum_{s\in\StateSpace}
\Prob[\State{t}=s\mid \State{0},\DeviationProfile{i}]
\CoarseAdvantage{i}{\Policy}{\DeviationPolicy{i}}{s}
\\
&=
\sum_{s\in\StateSpace}
\left(
\sum_{t=0}^{\infty}\Discount^t
\Prob[\State{t}=s\mid \State{0},\DeviationProfile{i}]
\right)
\CoarseAdvantage{i}{\Policy}{\DeviationPolicy{i}}{s}.
\end{aligned}
\]
By definition of the normalized discounted visitation distribution,
\[
\sum_{t=0}^{\infty}\Discount^t
\Prob[\State{t}=s\mid \State{0},\DeviationProfile{i}]
=
\frac{1}{1-\Discount}\VisitDist{\State{0}}{\DeviationProfile{i}}{s}.
\]
Substituting this identity yields \eqref{eq:coarse-perf-diff_}.

\end{proof}